\documentclass [letterpaper,10pt]{article}
\usepackage{notation}
\usepackage[left=1in,top=1in,right=1in,bottom=1.35in,letterpaper]{geometry}
\title{Low-rank Matrix Completion in a General Non-orthogonal Basis \thanks{This work is supported in part by NSF DMS--1522645 and an NSF Career Award DMS--1752934.}}
\author{ Abiy Tasissa 
             \thanks{Department of Mathematics, Rensselaer Polytechnic Institute, Troy, NY 12180, U.S.A. ({\tt tasisa@rpi.edu}).}         
              \and Rongjie Lai 
              \thanks{Department of Mathematics, Rensselaer Polytechnic Institute, Troy, NY 12180, 
             U.S.A. ({\tt lair@rpi.edu}).}
             }
              
        \date{}
\begin{document}
\maketitle

\begin{abstract}

This paper considers theoretical analysis of recovering a low rank matrix given a few expansion coefficients with respect to any basis.
The current approach generalizes the existing analysis for the low-rank matrix completion problem with sampling under entry sensing or with respect to a symmetric orthonormal basis. 
The analysis is based on dual certificates using a dual basis approach and does not assume the restricted isometry property (RIP). We introduce
a condition on the basis called the correlation condition. This condition can be computed in time $O(n^3)$ and holds for many cases of deterministic basis
where RIP might not hold or is NP hard to verify. If the correlation condition holds and the underlying low rank matrix obeys the coherence condition with parameter $\nu$,  under additional mild
assumptions, our main result shows that the true matrix can be recovered with very high probability from $O(nr\nu\log^2n)$ uniformly random expansion coefficients.

\end{abstract}

\section{Introduction}

Recovering low-rank matrices from given incomplete linear measurements plays an important role in many problems such as image and video processing~\cite{cai2017fast}, model reduction~\cite{fazel2001rank}, phase retrieval~\cite{candes2013phaselift}, molecular conformation~\cite{glunt1993molecular,trosset1997applications,fang2013using}, localization in sensor networks~\cite{ding2010sensor,biswas2006semidefinite}, dimensionality reduction~\cite{tenenbaum2000global}, recommender systems~\cite{kalofolias2014matrix} as well as solving PDEs on manifold-structured data represented as incomplete distance~\cite{lai2017solve}, just to name a few. 
A natural framework to the low-rank recovery problem is rank minimization under linear constraints.
However, this problem is NP-hard \cite{recht2010guaranteed} and thus motivates alternative solutions. 
A  series of theoretical papers \cite{candes2009exact,candes2010power,gross2011recovering,recht2011simpler,recht2010guaranteed}
showed that the NP-hard rank minimization problem for matrix completion  can be
obtained by solving the following convex nuclear norm minimization problem:
\begin{align}
\underset{\X \in \real^{n\times n}}{\textrm{minimize }} \quad &\|\X\|_{*} \nonumber\\
  \textrm{subject to }\quad &\X_{i,\,j} = \M_{i,\,j}\quad (i,j)\in \label{eq:nnm_generic}\Omega
\end{align}
where $\|\X\|_{*}$ denotes the nuclear norm defined as the sum of the singular values of $\X$, 
and $\Omega\subset\{(i,j)|i,j =1,...,n\}$, $|\Omega|=m$, denotes a random set 
that consists of the sampled indices. The remarkable fact is that, under certain conditions, the underlying low-rank matrix can be reconstructed exactly with high probability
from only $O(nr\log^{2}(n))$ uniformly sampled measurements. The idea to use the nuclear norm 
as an approximation of the rank function was first discussed in \cite{fazel2001rank}. 
Loosely, minimizing the sum of singular values will likely lead to a solution with many zero singular values resulting a low rank matrix.
One generalization of the matrix completion problem in \cite{gross2011recovering} considers measurements with respect to
a symmetric orthonormal basis and gives comparable theoretical guarantees based on the elegant dual certificate analysis. In particular, it shows that the true low rank matrix can be recovered with high probability from  
$O(nr\log^{2}(n))$ uniformly sampled measurements.

The starting point and inspiration for this work was our recent work in \cite{tasissa2018exact} which studies a matrix completion problem with respect 
to a specific non-orthogonal basis. In this paper, we consider the matrix completion problem with respect to any non-orthogonal basis. 
Given a general unit-norm basis $\{\w_{\alpha}\}_{\alpha=1}^{L}$
which spans an $L$ dimensional subspace $\S$ of $\real^{n\times n}$, 
the nuclear norm minimization program for this general matrix completion problem is provided by
\begin{align}
\underset{\X \in \S}{\textrm{minimize }} \quad &\|\X\|_{*} \nonumber\\
  \textrm{subject to }\quad &\langle \X\,,\w_{\alpha}\rangle  = \langle \M\,,\w_{\alpha}\rangle \quad \alpha \in \Omega
  \label{eq:nnm_minimization}
\end{align}
where $\Omega$ denotes a random set sampling the basis indices.  We are interested in the following two problems.
\begin{enumerate}
\item Could we obtain comparable recovery guarantees for the general matrix completion problem? 
\item If the answer to $(1)$ is affirmative, what conditions are needed on the basis $\{\w_{\alpha}\}$ and $\Omega$? 
\end{enumerate}
The main goal of this paper is a theoretical analysis of these two problems. We start by discussing few examples which show how the problem
naturally arises in several applications.

\paragraph{Euclidean Distance Geometry Problem.}

Given partial information on pairwise distances, the Euclidean distance geometry problem
is concerned with constructing the configuration of points. The problem has applications in diverse
areas \cite{glunt1993molecular,trosset1997applications,tenenbaum2000global,ding2010sensor,lai2017solve}. 
Formally, consider a set of $n$ points $\Pb=\{\p_{1},\p_{2},...,\p_{n}\}^{T} \in \real^{r}$.
Let $\D =[d_{i,j}^{2}]$ denote the Euclidean distance matrix. The inner product matrix,
also known as the Gram matrix and defined as $\X_{i,j} = \langle \p_{i}\,,\p_{j}\rangle$,
 is a positive semidefinite matrix of rank $r$.
A minor analysis reveals that $\D$ and $\X$ can be related in the following way: $\D_{i,j} = \X_{i,i}+\X_{j,j}-2\X_{i,j}$.
For $r\ll n$, consider the following nuclear norm minimization program to recover $\X$.
\begin{align}
\underset{\X \in \real^{n\times n}}{\textrm{minimize }} \quad &\|\X\|_{*} \nonumber\\
  \textrm{subject to }\quad &\X_{i,i}+\X_{j,j}-2\X_{i,j} = \D_{i,j}\quad (i,j)\in \Omega  \quad   
\label{eq:nnm_EDG_constraint}\\
   &\X\cdot \bm{1} = \bm{0} \nonumber \,;\,\X =\X^{T}\,;\, \X\succeq 0 \nonumber
\end{align}
The constraint $\X \cdot\bm{1} =\bm{0}$ fixes the
translation ambiguity.  The above minimization problem can be equivalently
interpreted as a general matrix completion problem with 
respect to some operator basis $\w_{\alpha}$.
\begin{align}
\underset{\X \in \Sp}{\textrm{minimize }} \quad &\|\X\|_{*} \nonumber\\
  \textrm{subject to }\quad &\langle \X\,,\w_{\alphab} \rangle = \langle \M\,,\w_{\alphab} \rangle \quad 
  \forall \alpha \in \Omega \label{eq:nnm_EDG_basis} 
\end{align}
where $\w_{\alphab} = \frac{1}{2}(\e_{\alpha_{1},\,\alpha_{1}}+\e_{\alpha_{2},\,\alpha_{2}}-\e_{\alpha_{1},\,\alpha_{2}}-\e_{\alpha_{2},\,\alpha_{1}})$
and $\Sp = \{\X\in \real^{n\times n} ~|~ \X=\X^{T}\,\&\,\X\cdot \bm{1}=\bm{0} \} \cap \{\X\in \real^{n\times n} ~|~ \X\succeq \bm{0}\}$. 
The constant $\frac{1}{2}$ is a normalization constant and $\e_{\alpha_1,\,\alpha_2}$ is a matrix whose entries are all zero except a $1$ at the $(\alpha_1,\,\alpha_2)$-th entry.
It can be verified that $\{\w_{\alphab}\}_{\alpha=1}^{L}$, $L = \frac{n(n-1)}{2}$, is a non-orthogonal basis for the linear space 
$\{\X\in \real^{n\times n} ~|~ \X=\X^{T}\,\&\,\X\cdot \bm{1}=\bm{0} \}$.  Theoretical analysis of this problem was recently conducted by the authors of this paper in~\cite{tasissa2018exact} 
and in fact inspires this work.

\paragraph{Spectrally Sparse Signal Reconstruction.}

The problem of signal reconstruction has many important practical applications.
When the underlying signal is assumed to be sparse, the theory of compressive sensing
states that the signal can be recovered by solving the convex $l_1$ minimization problem~\cite{candes2006robust}.
In \cite{cai2017fast},
the authors consider the recovery of spectrally sparse signal $\x\in \real^{n}$ of known order $r$
where $r\ll n$. Let $\mathcal{H}:\mathbb{C}^{n}\rightarrow \mathbb{C}^{n_1\times n_2}$
be the linear operator that maps a  vector $\bm{z} \in \mathbb{C}^{n}$ to a Hankel matrix $\mathcal{H}\bm{z} \in \mathbb{C}^{n_1\times n_2}$ with
$n+1 = n_1+n_2$. Denote the orthonormal basis of $n_1\times n_2$ Hankel matrices by $\{\H_{\alpha}\}_{\alpha=1}^{n_1\times n_2}$.
After some analysis, the reconstruction problem is formulated as low-rank Hankel matrix
completion problem \cite{cai2017fast}.
\[
\textrm{find} \quad \mathcal{H}\bm{z} \quad \textrm{ subject to } \quad \textrm{rank}(\mathcal{H}\bm{z}) = r \quad \P_{\Omega}(\mathcal{H}\bm{z}) =
\P_{\Omega}(\mathcal{H}\bm{x})
\]
$\P_{\Omega}$ is the sampling operator defined as: $\P_{\Omega}(\bm{Z}) = \sum_{\alpha \in \Omega} \langle \Z\,,\bm{H}_{\alpha}\rangle \bm{H}_{\alpha}$
where  $\Omega$ is the random set that consists of the sampled Hankel basis. 
For the case where $r$ is not specified, one can consider the following nuclear norm minimization problem.
\begin{align}
\underset{\X \in \real^{n_1\times n_2}}{\textrm{minimize }} \quad &\|\X\|_{*} \nonumber\\
  \textrm{subject to }\quad &\langle \X\,,\H_{\alpha} \rangle = \langle \M\,,\H_{\alpha}\rangle \quad 
  \forall \alpha \in \Omega \label{eq:nnm_spectrally_sparse}
\end{align}
where $\X = \mathcal{H}\bm{z}$ and $\M = \mathcal{H}\x$. \eqref{eq:nnm_spectrally_sparse}
is now in the form of a general matrix completion problem.

\paragraph{Signal Recovery Under Quadratic Measurements.}

Given a signal $\x \in \real^{n}$, consider quadratic measurements of the form $\bm{a}_{i} = \langle \x\,, \bm{w}_{i}\rangle^{}$
with some vector $\bm{w}_{i}\in \real^{n}$. Given that $m$ random measurements are available, it is of interest to
determine if the underlying signal can be recovered.
Using the lifting idea in \cite{candes2015phase}, the signal recovery problem can be written
as the following nuclear norm minimization problem.
\begin{align}
\underset{\X \in \real^{n\times n}}{\textrm{minimize} } \quad &\|\X\|_{*} \nonumber\\
  \textrm{subject to }\quad &\langle \X\,,\W_{\alpha} \rangle = \langle \M\,,\W_{\alpha} \rangle \quad 
  \forall \alpha \in \Omega\\
     &\,\X =\X^{T}\,;\, \X\succeq 0 \nonumber
\end{align}
Above $\W_{\alpha} = \bm{w}_{\alpha}\bm{w}_{\alpha}^{*}$ and $\Omega$ is the random set that consists
of the indices of the sampled vectors. In the case that $\bm{a}_{i}$ are sampled independently and uniformly at random
on the unit sphere, the above minimization problem is the well known PhaseLift problem~\cite{candes2013phaselift}.
One can consider a general case where the assumption is simply that the $\bm{w}_{i}$'s are structured and
form a basis. The framework introduced
in this paper allows, under certain conditions, to state results about
the uniqueness of this general recovery problem.

\paragraph{Weighted Nuclear Norm minimization.}

The usual assumption in matrix completion is uniform random measurements.
In the case of general sampling models, it has been argued that the nuclear norm minimization
is not a suitable approach \cite{srebro2010collaborative}.
An alternative which has been argued to promote more accurate low rank solutions \cite{candes2015phase,fazel2003log} is the weighted nuclear norm minimization \cite{foygel2011learning,srebro2010collaborative}. With $\D$ as a weight matrix, the weighted nuclear norm minimization problem is given by
\begin{align}
\textrm{minimize } \quad &\|\D\X\|_{*} \nonumber\\
\textrm{subject to }\quad &\langle \X\,,\w_{\alpha} \rangle = \langle \M\,,\w_{\alpha} \rangle \quad 
 \forall \alpha \in \Omega 
\end{align}
For simplicity, let $\D$ be a diagonal matrix. In this case, $\|\D\X\|$ can be interpreted as weighting certain rows of $\X$ more than others. 
Introducing the matrix $\Y = \D\X$, the above minimization
problem can equivalently be rewritten as follows.
\begin{align}
\textrm{minimize } \quad &\|\Y\|_{*} \nonumber\\
 \textrm{subject to }\quad &\langle \Y\,,\D^{-1}\w_{\alpha} \rangle = \langle \overline{\M}\,,\D^{-1}\w_{\alpha} \rangle \quad 
\forall \alpha \in \Omega 
\end{align}
$\overline{\M}$ is the weighted ground matrix $\M$ defined as $\overline{\M}= \D \M$. This is a general matrix completion problem with respect to the basis
$\D^{-1} \w_{\alpha}$. Since $\D$ is a diagonal matrix, the set $\{\D^{-1} \w_{\alpha}\}$ is a basis. Of interest
is the following question: What kind of choices for $\D$ lead to successful recovery algorithms?

\paragraph{Challenges}

The minimization problem in \eqref{eq:nnm_minimization} has random linear constraints which can be expressed as $\mathcal{L}(\X)$
where $\mathcal{L}$ is the appropriate linear operator. Using the work in \cite{recht2010guaranteed}, one approach
to show uniqueness of the general matrix completion problem is to check if $\mathcal{L}$ obeys
the restricted isometry property (RIP) condition. Our basis are structured and deterministic
and so in general the RIP condition does not hold. As an example, consider the nuclear
norm minimization program in \eqref{eq:nnm_EDG_basis} for the Euclidean distance geometry
problem. We choose any $(i,j)\notin \Omega$ and construct a matrix $\X$ with $\X_{i,j} = \X_{j,i} = \X_{i,i}=\X_{j,j}=1$
and zero everywhere else. One can easily check that $\mathcal{L}(\X) = \bm{0}$ which shows that the RIP
condition does not hold.  The general matrix completion problem resembles the matrix completion problem with respect to a symmetric orthonormal basis
first considered in \cite{gross2011recovering}. However, the basis $\{\w_{\alpha}\}_{\alpha=1}^{L}$ in the general
matrix completion problem are not necessarily orthogonal. This has the implication that the measurements $\langle \w_{\alpha},\X \rangle$ are not
compatible with the expansion coefficients of $\M$. One solution is to employ any of the basis orthogonalization
algorithms and consider the minimization problem in the new orthonormal basis. This solution
however is not useful since the measurements can not be treated as independent
in the new basis. As such, the lack of orthogonality mandates an alternative analysis to
show that the general matrix completion problem admits a unique solution. In this paper, the analysis is based on the dual certificate approach \cite{candes2009exact}.
Motivated by the work of David Gross \cite{gross2011recovering},
where the author generalizes the matrix completion problem to any symmetric orthonormal basis, 
our recent work in \cite{tasissa2018exact} considered the Euclidean distance geometry problem \eqref{eq:nnm_EDG_basis}.
The main technical difference from the work of \cite{gross2011recovering} is that the basis in the Euclidean distance geometry problem is
non-orthogonal. More precisely, in terms of analysis, the difference is mainly due to the sampling operator, 
which is central in the analysis of the matrix completion problem. For the orthonormal basis case, the sampling operator
is self-adjoint. For the Euclidean distance geometry problem, the sampling operator is not self-adjoint
and requires alternative analysis.  Much inspired by our previous work, we are interested in generalizing the result to any non-orthogonal basis. 
The current paper is a culmination of this effort. The work in \cite{kueng2014ripless} develops RIPless recovery analysis for the compressed sensing problem from anisotropic
measurements. The current work could be interpreted as analogue of this work to the general matrix completion problem.  In particular,
the notion of anistropic measurements for compressive sensing corresponds to non-orthogonal
matrix basis for matrix completion.  There are however differences as a direct analogue of some of the technical estimates in \cite{kueng2014ripless}
is not directly applicable to the general matrix completion problem.

\paragraph{Contributions}
In this paper, under suitable sampling conditions, a dual basis approach is used to show that the general matrix completion problem
admits a unique solution. Introducing
a dual basis to
$\w_{\alpha}$, denoted by $\z_{\alpha}$, ensures that the measurements $\langle \X\,,\w_{\alpha}\rangle$
in \eqref{eq:nnm_minimization} are compatible with expansion coefficients
of $\M$. 
Based on the framework of the dual basis approach, we show that the minimization problem
recovers the underlying matrix under suitable conditions. Two main contributions of this paper are as follows.
\begin{enumerate}
  \item A dual basis approach is used to prove a uniqueness result for the general matrix completion problem. 
  The main result shows that if the number of random measurements $m$ is of order $O(nr\log^{2}n)$, 
  under certain assumptions, the nuclear norm minimization program recovers the underlying low-rank solution
  with very high probability. A key part of our proof uses the operator Chernoff bound. This part of the proof
based on the Chernoff bound is simple and might find use in other problems.

\item An important condition, named the correlation condition, is introduced. This condition determines whether the nuclear norm minimization program
  succeeds for a given general matrix completion problem. The well-known RIP condition might not hold for the case of deterministic measurements. However, 
  the correlation condition could hold for deterministic measurements  and more importantly can be checked in polynomial time.  

 \end{enumerate}

\paragraph{Outline}
The outline of the paper is as follows. Section \ref{sec:CorrelationCond} introduces the correlation condition
and discusses the dual basis approach. The proof of the main result is presented in
section \ref{sec:mainresult}. The key components of the proof can be described as follows. 
The general matrix completion problem is a convex minimization problem for which
a sufficient condition to optimality is the KKT condition. Namely, if one
can show that there exists a dual certificate $\Y$ that satisfies certain conditions,
it follows that the general matrix completion problem has a unique solution. 
The construction of $\Y$ follows the elegant golfing scheme proposed in \cite{gross2011recovering}. 
The bulk of the theoretical work then focuses on using this scheme and 
proving that the conditions hold with very high probability. The implication
of the proof is that there is a unique solution to the general matrix completion
problem with very high probability. Section \ref{sec:conclusion} concludes the work.

\paragraph{Notation}

The notations used in the paper are summarized in Table \ref{tab:notation}.
\begin{table}[h!]
\centering
\begin{tabular}{|cc||cc|}
\hline
$\x$ & Vector & $\|\X\|_{F}$  &  Frobenius norm\\
$\X$ & Matrix & $\|\X\|_{\infty}$ & $ \sup_{\|v\|_\infty=1}\|\X v\|_\infty$\\
$\mathcal{X}$ & Operator & $\|\X\|$ &  $ \sup_{\|v\|_2=1}\|\X v\|_2$\\ 
$\X^T$ & Transpose & $\|\X\|_{*}$ & Nuclear norm\\
$\textrm{Tr}(\X)$ & Trace & $\|\mathcal{A}\|$ & $\sup_{\|\X\|_{F}=1}\|\mathcal{A}\X\|_{F}$.\\
$\langle \X\,,\Y\rangle$ & $\textrm{Trace}(\X^T\Y)$ & $\lambda_{\max},\lambda_{\min}$ & Maximum, Minimum eigenvalue\\
$\bm{1}$ & A vector or matrix of ones & $\textrm{Sgn }\X$ & $\U\textrm{sign }(\Sigma) \V^T$; here $[\U,\Sigma,\V]= \text{svd}(\X)$\\
$\bm{0}$ & A vector or matrix of zeros & $\Omega$, $\mathbb{I}$ & Random sampled set, Universal set\\\hline
\end{tabular}
\caption{Notations}
\label{tab:notation}
\end{table}

\section{Matrix Completion Problem under a Non-orthogonal Basis}
\label{sec:CorrelationCond}
In this section, we introduce a new condition referred as a {\it correlation condition} for low-rank matrix completion in a non-orthogonal basis. We will also discuss a dual basis formulation which plays an important role in our main result. 

\subsection{Correlation Parameter}
We consider an $L$ dimensional subspace $\S\subset\real^{n\times n}$ as the feasible set of the general matrix completion. 
We allow $L \leq n^2$ such as the case
where the feasible solutions naturally satisfy linear constraints. 
Given a complete set of unit-norm basis $\{\w_{\alpha}\}_{\alpha=1}^L$ of the subspace $\S$,
any $\X\in \S$ can be determined if one specifies all the measurements $\{\langle \X\,,\w_{\alpha}\rangle\}_{\alpha=1}^{L}$. 
The problem studied in this paper considers the case where we have random access to a few of these measurements.  
Leaving the precise notion of ``a few'' for later, we consider the following question: Are there non-orthogonal basis for which the matrix completion framework, learning from few measurements, still works?
In this paper, we note that as long as a certain condition, named correlation condition, on the basis matrices is satisfied, the sample
complexity of low rank matrix completion with respect to non-orthogonal basis is of the same order as sample
complexity of low rank matrix completion with respect to an orthogonal basis.  
Intuitively, there is decoupling in orthogonality which means that every additional measurement is informative. 
However, with non-orthogonal basis, the basis might be correlated and every additional measurement might not necessarily be informative.
In the specific case of the general matrix completion
problem, the goal is to rigorously show that, under certain conditions, a low rank matrix can be recovered from few
random non-orthogonal measurements. The intuitive arguments above motivate the following correlation condition
which loosely informs how far the nonorthogonal basis is from an orthonormal basis. 

\begin{definition}
The unit-norm basis $\{\w_\alpha\in\real^{n\times n}\}_{n=1}^{L}$ has correlation parameter $\mu$ if there is a constant $\mu\ge 0$
such that the following two equations hold
\begin{equation}\label{eq:correlation}
\left\|\frac{1}{n}\sum_{\alpha} \w_{\alpha}^{T} \w_{\alpha} -\I \right\|\le \mu \quad \& \quad 
\left\|\frac{1}{n}\sum_{\alpha} \w_{\alpha} \w_{\alpha}^{T}-\I \right\|\le \mu
  \end{equation}
\end{definition}
Intuitively, the above definition describes that the operators $\frac{1}{n}\sum_{\alpha} \w_{\alpha}^{T} \w_{\alpha}$ and
$\frac{1}{n}\sum_{\alpha} \w_{\alpha} \w_{\alpha}^{T}$ are nearly isometric to the identity operator.  
To understand the correlation condition, we first establish certain properties and follow by computing the correlation condition of certain basis
matrices. 
\begin{lemma} \label{corr_props} [Properties and examples of correlation condition]
Given a unit-norm basis $\{\w_\alpha\in\real^{n\times n}\}_{n=1}^{L}$ with correlation parameter $\mu$, the following statements hold:

\noindent
\begin{enumerate}[label=\alph*.]

\item  The correlation parameter is bounded above by $n$, that is, $\mu\le n$.

\item If the correlation condition holds, it follows that 
  \[
\lambda_{\max}\left(\sum_{\alpha=1}^{L} \w_{\alpha} ^{T}\w_{\alpha}\right) \le  (\mu+1) n \quad\& \quad
 \lambda_{\max}\left(\sum_{\alpha=1}^{L} \w_{\alpha} \w_{\alpha}^{T}\right) \le (\mu+1) n 
\]

\item If $\{\w_\alpha\}_{n=1}^{L}$ is an orthonormal basis $L = n^2$, the correlation condition holds with $\mu=0$. 

\item For the basis matrices in the Euclidean distance geometry problem, the
correlation condition holds with $\mu = 1$. 

\end{enumerate}
\end{lemma}

\begin{proof}

\noindent
\begin{enumerate}
\item For positive semidefinite matrices $\bm{A}$ and $\bm{B}$, the norm inequality $\|\bm{A}-\bm{B}\|\le \max(\|\bm{A}\|,\|\bm{B}\|)$ holds.
Let $\bm{A} = \frac{1}{n}\sum_{\alpha} \w_{\alpha}^{T} \w_{\alpha}$ and let $\bm{B} = \I$. Applying the
triangle inequality, it follows that
$\left\| \frac{1}{n}\sum_{\alpha} \w_{\alpha}^{T} \w_{\alpha} \right\|\le \frac{1}{n}\sum_{\alpha} ||\w_{\alpha}^{T}||\, ||\w_{\alpha}|| \leq n$.
Therefore, $\left\|\frac{1}{n}\sum_{\alpha} \w_{\alpha}^{T} \w_{\alpha} -\I \right\|\le \max(n,1) = n$.
An analogous argument obtains $\left\|\frac{1}{n}\sum_{\alpha} \w_{\alpha} \w_{\alpha}^{T}-\I \right\|\le n$. 
Hence, $\mu \le n$ as desired. It should be remarked that, although $\mu$ is bounded above by $n$, the theoretical analysis requires
$\mu = O(1)$. 

\item Note that 
$\lambda_{\max}\left(\sum_{\alpha=1}^{L} \w_{\alpha} ^{T}\w_{\alpha}\right)  =  \|\sum_{\alpha=1}^{L} \w_{\alpha} ^{T}\w_{\alpha}-n\I+n\I  \| \le \|\sum_{\alpha=1}^{L} \w_{\alpha} ^{T}\w_{\alpha}-n\I \|+ \| n\I \| =(\mu+1)n$. A similar argument results $ \lambda_{\max}\left(\sum_{\alpha=1}^{L} \w_{\alpha} \w_{\alpha}^{T}\right) \le (\mu+1) n $.

\item For an orthonormal basis in $\S = \real^{n\times n}$, the completeness relation states that  $\sum_{\alpha=1}^{n^2} \w_{\alpha} \w_{\alpha}^{T} = \sum_{\alpha=1}^{n^2} \w_{\alpha}^{T} \w_{\alpha}  = n\I$.
The proof of this fact is standard. For ease of reference, a proof is included in Lemma \ref{completeness}.
Using the completeness relation, it follows that the correlation condition holds with $\mu = 0$. 

\item For the Euclidean distance geometry problem, the basis are symmetric so it suffices to consider $\frac{1}{n} \sum_{\alpha} \w_{\alpha}^{2}$.
A short calculation results  $\frac{1}{n} \sum_{\alpha} \w_{\alpha}^{2} = \frac{1}{2n}(n\I-\bm{1}\bm{1}^{T})$. The correlation condition bound
amounts to finding the operator norm of $\frac{1}{2} \|\I+\frac{1}{n}\bm{1}\bm{1}^{T}\|$. Since the maximum eigenvalue of $\frac{1}{n}\bm{1}\bm{1}^{T}$
is $1$,  $\frac{1}{2} \| \I+\frac{1}{n}\bm{1}\bm{1}^{T} \| = 1$. Hence, for the Euclidean distance geometry problem, the
correlation condition holds with $\mu = 1$. 

\end{enumerate}

\end{proof}

\paragraph{Comparison with RIP condition}

The restricted isometry property, RIP in short, was introduced first in \cite{candes2005decoding} for the compressive sensing problem.
If the measurement matrix in the compressive sensing problem, denoted by $\bm{A}$, obeys the RIP condition, it implies that the underlying sparse vector can be recovered
by solving a convex $l_1$ minimization problem. The notion of the restricted isometry property can be extended to the matrix
completion problem and its definition, which first appeared in \cite{recht2010guaranteed}, is restated below. 
\begin{definition}

  For every integer $r$ with $1\le r \le n$,  a linear map $\A: \real^{n\times n}\rightarrow \real^{m}$ obeys the RIP condition
  with isometry constant $\delta_{r}$ if
  \[
  (1-\delta_r)\|\X\|_{F}^{2}\le \|\A\X\|_{2}^{2} \le (1+\delta_r)\|\X\|_{F}^{2}
  \]
  holds for all matrices $\X$ of rank at most $r$. 
\end{definition}

In particular, the analysis in \cite{recht2010guaranteed} shows that if the RIP condition holds for the measurement operator
of the general matrix completion problem, the underlying low rank matrix can be recovered by solving the convex
nuclear norm minimization problem. With this, one could only consider certifying that the measurement operator
satisfies RIP.  For example,  the RIP condition is satisfied with high probability for random measurement operators  \cite{recht2010guaranteed}. 
RIP condition could fail if one considers deterministic measurement operators. The intuition is that, for this case,
one could in some fashion construct matrices which belong to the null space of $\A$ implying that RIP does not hold.
As noted earlier, the RIP condition fails to hold for the EDG problem where the proof of this fact relied on constructing a counterexample, a sparse matrix $\X$, which violates the RIP condition.
On the other hand,  Lemma \ref{corr_props} $(4)$ shows that the correlation condition holds for the EDG problem
with correlation parameter $\mu =1$.
Additionally, one could construct counter examples that imply that the RIP condition does not hold for the standard matrix completion problem and
the phase retrieval problem. Therefore, in certain deterministic settings where the RIP condition does not hold,
an RIPless analysis based on correlation parameter can be employed. Another advantage of the correlation condition
is its computational complexity. Given a measurement operator, checking whether RIP condition holds or not is hard. For example,
for the compressive sensing problem, it has been shown that \cite{bandeira2013certifying} checking whether RIP holds
or not is NP-hard.  To the best of our knowledge, there is no analogous work for the matrix RIP condition. 
However, using the relations of vector RIP to matrix RIP as detailed in \cite{oymak2011simplified},
it can be anticipated that checking the matrix RIP is also NP-hard.  
On the other hand,  the computational complexity of checking the correlation condition is at most $O(n^3)$. 
As such, given a measurement operator, it can easily be checked whether it satisfies the correlation condition or not. If the
result is positive, an RIPless analysis can be carried out as will be detailed in this paper. If the correlation condition does
not hold, no conclusion can be made. 
The above comparison is meant to illustrate that, for the case of deterministic measurements, the correlation condition
could be used and an RIPless analysis can be carried out for the general matrix completion problem.

\subsection{Dual basis formulation}
For the general matrix completion problem in \eqref{eq:nnm_minimization}, the measurements are in the form $\langle \X\,,\w_{\alpha}\rangle $. 
If $\w_{\alpha}$ is orthonormal, $\X$ can be expanded as $\X = \sum_{\alpha} \langle \X\,,\w_{\alpha}\rangle \w_{\alpha}$
where $\langle \X\,,\w_{\alpha}\rangle$ are the expansion coefficients. However, for non-orthogonal basis, this expansion does not hold. 
Since the measurements are inherent to the problem, the ideal expansion will be of the form $\X = \sum_{\alpha}  \langle \X\,,\w_{\alpha}\rangle \z_{\alpha}$.
The idea of the dual basis approach is to realize this form with the implication that the expansion coefficients match the random measurements.
This approach is briefly summarized below and we refer the interested reader to \cite{tasissa2018exact} where the dual
basis approach is first considered in the context of the Euclidean distance geometry problem. Given 
the basis $\{\w_{\alpha}\}_{\alpha=1}^{L}$ of $\S$,  we define the matrix $\H$ as 
$\H_{\alpha,\,\beta} = \langle \w_{\alpha}\,,\w_{\beta}\rangle$ and write $\H^{-1}_{\alpha,\,\beta}$ as $\H^{\alpha,\,\beta}$. 
After minor analysis, it is straightforward to check $\z_{\alpha}= \sum_{\beta} \H^{\alpha,\,\beta} \w_{\beta}$ is a dual basis satisfying
$\langle \z_{\alpha}\,,\w_{\beta} \rangle = \delta_{\alpha,\,\beta}$. 
We note the following relations which will be used in later analysis, $\H = \W^T\W$, $\H^{-1} = \Z^T\Z$ and $\Z = \W\H^{-1}$
where $\W =[\w_1,\w_2,...,\w_L]$ and $\Z = [\z_1,\z_2,...,\z_{L}]$ denote the matrix of vectorized basis matrices 
and vectorized dual basis matrices respectively. The sampling operator, a central operator in the
analysis of the general matrix completion problem, is defined as follows.
\begin{equation} \label{eq:dual_sampling}
\Po:\X \in \S\longrightarrow \frac{L}{m}\sum_{\alpha\in\Omega}\langle \X\,,\w_{\alpha} \rangle \z_{\alpha}
\end{equation}
The sampling model, for the set $\Omega$, is uniform random with replacement from $\Us = \{1,\cdots,L\}$. The size of $\Omega$ is denoted by $m$  and the scaling factor $\displaystyle \frac{L}{m}$ is simply for convenience of analysis.
The adjoint operator of the sampling operator also
appears in the analysis and has the following form.
\begin{equation} 
\Po^{*}:\X \in \S\longrightarrow \frac{L}{m}\sum_{\alpha\in\Omega} \langle \X\,,\z_{\alpha} \rangle \w_{\alpha}
\end{equation}
Using the sampling operator $\Po$, we can write \eqref{eq:nnm_minimization} as follows. 
\begin{align}
&\underset{\X\in \S}{\textrm{minimize }} \quad \|\X\|_{*} \nonumber\\
&\textrm{subject to }\quad \Po(\X) = \Po(\M)  \label{eq:m_completion}
\end{align}
Another operator which appears in the analysis is the restricted frame operator defined as follows.
\begin{equation} \label{eq:frame_operator}
\F:\X \in \S \longrightarrow \frac{L}{m}\sum_{\alphab\in\Omega}\langle \X\,,\w_{\alphab} \rangle \w_{\alphab}
\end{equation}
It can be readily verified that the restricted frame operator is self-adjoint and positive semidefinite.

\paragraph{Coherence.}
One can not expect to have successful reconstruction for an arbitrary matrix $\M$, in particular, when $\M$ has very few non-zero expansion coefficients. 
Thus,  the notion of coherence is introduced in~\cite{candes2009exact} to guarantee successful completion. Consider the singular value decomposition of $\M = \sum_{k=1}^{r} \lambda_{k} \u_{k} \v_{k}^{T}$.
Let's write $\U = \textrm{span } \{\u_1,...,\u_r\}$,
$\U^{\perp} = \textrm{span} \{\u_{r+1} , . . . , \u_n \}$ as the orthogonal complement of $\U$,
$\V = \textrm{span } \{\v_1,...,\v_r\}$ and $\V^{\perp} = \textrm{span} \{\v_{r+1} , . . . , \v_n \}$ as the orthogonal complement of $\V$.
Projections onto these spaces appear frequently in the analysis and can be summarized
as follows. Let $\P_{\U}$ and $\P_{\V}$ denote the orthogonal projections
onto $\U$ and $\V$ respectively. $\P_{\U^{\perp}}$ and $\P_{\V^{\perp}}$ are defined analogously. 
Define $\T=\{\U\Theta^T+\Theta\V^{T}:\Theta\in \real^{n\times r} \}$ to be the tangent space of the rank $r$ matrix in $\real^{n\times n}$ at $\M$.
The orthogonal projection onto $\T$ is given by
\begin{equation} \label{eq:pt_defn}
\P_{\T}\X = \P_{\U}\X + \X \P_{\V} - \P_{\U} \X \P_{\V}
\end{equation}
It then follows that $\P_{\T^{\perp}}\X = \X - \P_{\T}\X = \P_{\U^{\perp}}\X\P_{\V^{\perp}}$.
We define a coherence condition as follows. 
\begin{definition}
The aforementioned rank r matrix $M\in \real^{n\times n}$ has coherence $\nu$ with respect to unit norm basis $\{\w_{\alpha}\}_{\alpha\in \Us}$ if
the following estimates hold
\begin{align}
\underset{\alpha\in \Us}{\max} \,\, &\sum_{\beta\in \Us}\langle \P_{\T}\, \w_{\alpha}\,,\w_{\beta} \rangle^{2} \le \nu \frac{r}{n}
\label{eq:coherencew1}\\
\underset{\alpha\in \Us}{\max} \,\, &\sum_{\beta\in \Us}\langle \P_{\T}\, \z_{\alpha}\,,\w_{\beta} \rangle^{2} \le c_{v}\nu \frac{r}{n}
\label{eq:coherencev1}\\
\underset{\alpha\in \Us}{\max} \,\, &\langle \w_{\alpha}\,,\U\V^{T}\rangle^{2} \le \nu \frac{r}{n^{2}} \label{eq:jointcoherence1}
\end{align}
where $\{\z_{\alpha}\}_{\alpha\in \Us}$ is the dual basis of $\{\w_{\alpha}\}_{\alpha\in \Us}$
and $c_v$ is a constant satisfying $c_{v}\ge \lambda_{\max}(\H^{-1})\|\H^{-1}\|_{\infty}$.
\end{definition}

\noindent
\begin{remark}

Since the above coherence conditions are a central part of the analysis, we consider equivalent simplified forms.
We start with a bound on  $\|\P_{\T}\, \w_{\alpha}\|_{F}^{2}$ using \eqref{eq:coherencew1} and 
Lemma \ref{norm_equivalent} which results 
\begin{equation*} 
\lambda_{\min}(\H)\,\|\P_{\T}\, \w_{\alpha}\|_{F}^{2}\le 
\underset{\alpha\in \Us}{\max} \,\, \sum_{\beta\in \Us}\langle \P_{\T}\, \w_{\alpha}\,,\w_{\beta} \rangle^{2} \le \nu \frac{r}{n}\quad 
\Longrightarrow \quad \|\P_{\T}\, \w_{\alpha}\|_{F}^{2}\le \lambda_{\max}(\H^{-1})\,\nu \frac{r}{n}
\end{equation*}
Next,  using the above inequality, the following bound for $\|\P_{\T}\, \z_{\alpha}\|_{F}$ follows.
\begin{equation*}
\|\P_{\T}\, \z_{\alpha}\|_{F} \le \sum_{\beta\in \Us} \|\H^{\alpha,\,\beta}\,\P_{\T}\, \w_{\beta}\|_{F} = 
\sum_{\beta\in \Us} |\H^{\alpha,\,\beta}|\,\,\|\P_{\T}\, \w_{\beta}\|_{F} \le \|\H^{-1}\|_{\infty}\,\sqrt{\lambda_{\max}(\H^{-1})} \sqrt{\frac{\nu r}{n}}                   
  \end{equation*}
It should be noted that the analysis presented in this paper requires that $\|\H^{-1}\|_{\infty}$ is at most $O(1)$.
Finally, we use the previous inequality and \eqref{eq:jointcoherence1} to derive a bound for $\langle \z_{\alpha}\,,\U\V^{T}\rangle^{2} $.
\begin{align*}
 \langle \z_{\alpha}\,,\U\V^{T}\rangle^{2}	= \langle \sum_{\beta\in \Us} \H^{\alpha,\,\beta}\w_{\beta}\,,\U\V^{T}\rangle^{2}	&\le
\underset{\beta\in \Us}{\max} \,\, \langle \w_{\beta}\,,\U\V^{T}\rangle^{2}\,\left(\sum_{\beta\in \Us} |\H^{\alpha,\,\beta}|\right)^{2} \\
&\le \|\H^{-1}\|_{\infty}^{2}\, \frac{\nu r}{n^2} 
\end{align*}
The coherence conditions can now be summarized as follows.
\begin{align}
\underset{\alpha\in \Us}{\max} \,\, &\|\P_{\T}\, \w_{\alpha}\|_{F}^{2} \le \lambda_{\max}(\H^{-1})\,\nu \frac{r}{n}\label{eq:coherencew}\\
\underset{\alpha\in \Us}{\max} \,\, &\|\P_{\T} \, \z_{\alpha}\|_{F}^{2} \le \|\H^{-1}\|_{\infty}^{2}\,\lambda_{\max}(\H^{-1})\,\nu \frac{r}{n}\label{eq:coherencev}\\
\underset{\alpha\in \Us}{\max} \,\, &\langle \z_{\alpha}\,,\U\V^{T}\rangle^{2} \le \|\H^{-1}\|_{\infty}^{2}\, \frac{\nu r}{n^2} \label{eq:jointcoherence}
\end{align}
\end{remark}

\noindent
The coherence parameter is indicative of concentration of information
in the ground truth matrix. If the underlying matrix has low coherence, each measurement is equally
informative as the other. On the other hand, if a matrix has high coherence, it means that the information
is concentrated on few measurements.

\paragraph{Sampling Model.} For the general matrix completion problem, the basis matrices
are sampled uniformly at random with replacement. The advantage of this model is that the
sampling process is independent. This property is crucial since our analysis 
uses concentration inequalities for i.i.d matrix valued random variables.
A disadvantage of this sampling process is that the same measurement could be repeated
and the analysis needs to account for the number of duplicates.

\begin{remark}

Although we choose uniform sampling with replacement model, the analysis in this paper
also works for sampling with out replacement. The latter model has the 
advantage that there are no duplicate measurements but the choice also means that
the sampling is no longer independent. This in turn has the implication that
concentration inequalities for i.i.d matrix valued random
variables can not be used freely.  However, in the work of Hoeffding \cite{hoeffding1963probability}, for Hoefdding inequality,
it is argued that the results derived for the case of the sampling with replacement
also hold true for the case of sampling without replacement.  In \cite{gross2010note},
it is shown that matrix concentration inequalities resulting from  the operator Chernoff bound technique
\cite{ahlswede2002strong} also hold true for uniform sampling without replacement. 
With this, the main analysis in this paper holds with or without replacement. 
The use of uniform sampling with out replacement model in the analysis leads to a  gain in terms of the number of measurements $m$. However, this
gain is rather minimal and for sake of streamlined
presentation, the uniform sampling with replacement is adopted in this paper.

\end{remark}

\section{Main Result and proof}
\label{sec:mainresult}

The main result of this paper shows that the nuclear norm minimization program for the general matrix completion problem 
in \eqref{eq:m_completion} recovers the underlying matrix with very high probability. A precise statement is stated in the theorem below. 
With out loss of generality and for ease of analysis, the theorem considers square matrices. 
The proof for rectangular matrices follows with minor modifications. 
\begin{theorem}
Let $\M\ \in \real^{n\times n}$ be a matrix of rank $r$ that obeys the coherence conditions
\eqref{eq:coherencew1}, \eqref{eq:coherencev1} and \eqref{eq:jointcoherence1} with coherence 
$\nu$ and satisfies the correlation condition \eqref{eq:correlation} with correlation parameter $\mu$.
Define $C$ as follows: 
$\displaystyle C = \max\left(\lambda_{\max}(\H^{-1})^{3}, c_{v}, \frac{(\mu+1) \|\H^{-1}\|_{\infty} }{\min(\,(\mu+1)\|\H^{-1}\|_{\infty},\frac{1}{4})^{2}}\right)$
with parameter $c_{v}$ from \eqref{eq:coherencev1}.
Assume $m$ measurements, $\{\langle \M\,,\w_{\alpha}\rangle\}_{\alpha\in\Omega}$, are sampled uniformly at random with replacement.
For $\beta>1$, if
\begin{equation} \label{eq:m_general_basis}
m \ge \log_{2}\left(4\sqrt{2L}\frac{\lambda_{\max}(\H)}{\lambda_{\min}(\H)}\sqrt{r}\right)nr
\left(48\big[C\nu+\frac{n}{Lr}\big]\big[\,\beta\log(n)+\log\left(4\log_{2}\left(4\sqrt{2L}\frac{\lambda_{\max}(\H)}{\lambda_{\min}(\H)}\sqrt{r}\right)\right)\big]\right)
\end{equation}
the solution to \eqref{eq:m_completion} is unique and equal to $\M$ with probability at least $1-n^{-\beta} $. 
\label{thm:thm1}
\end{theorem}

Since the general matrix completion problem in \eqref{eq:m_completion} is convex, the optimizer can be characterized using the KKT
conditions.  A compact and simple form of these conditions is derived in \cite{candes2009exact}.
With this, a brief outline of the proof is as follows. The proof is divided into two main parts. 
In the former part, we show that if the aforementioned optimality  conditions hold, 
then $\M$ is a unique solution to the minimization problem. The latter and main part of the proof is concerned with showing that, under certain assumptions, these conditions do hold with very high probability.
The implication of this is that, for a suitable choice of $m$, $\M$ is a unique solution for the general matrix completion problem.

Our proof adapts arguments from \cite{gross2011recovering,tasissa2018exact}. Few remarks on the difference of
our proof to the matrix completion proofs in \cite{gross2011recovering,recht2011simpler} and our previous work \cite{tasissa2018exact} are in order.
\begin{enumerate}
\item The operator $\Po$ is not self-adjoint. The main implication of this is that the operator $\P_{\T}\Po\P_{\T}$, an important
operator in matrix completion analysis, is no longer isometric to $\P_{\T}$. It turns out the appropriate operator to consider is
$\P_{\T}\Po^{*}\Po\P_{\T}$ where the goal is to show that this operator is nearly isometric to $\P_{\T}$.
However, this approach is not amenable to simple analysis. In this work, the main argument is
based on showing that the minimum eigenvalue of the operator $\P_{\T}\F\P_{\T}$ is bounded away from zero
with very high probability. To prove this fact, the operator Chernoff bound is employed.  
The interpretation of this bound is that, restricted to the space $\T$, the operator $\P_{\T}\F\P_{\T}$
is full rank. If the measurement basis is orthogonal, $\F = \Po$, the implication is that
the operator $\P_{\T}\Po\P_{\T}$ on $\T$ is invertible.  With this, $\P_{\T}\F\P_{\T}$ can be understood 
as the operator analogue of $\P_{\T}\Po\P_{\T}$ for non-orthogonal measurements.

\item The measurement basis is non-orthogonal. Since we use the dual basis approach, the spectrum of the matrices $\H$ and $\H^{-1}$
become important. However, since we do not work with a fixed basis, all the constants are unknown. This is particularly relevant
and presents some challenge in the use of concentration inequalities and will be apparent in later analysis. 
\end{enumerate}
For the matrix completion problem, with measurement basis
$\bm{e}_{ij}$, the theoretical lower bound of $O(nr\nu\log\,n)$ was established in \cite{candes2010power}.
Note that, if $C=O(1)$ and $\mu = O(1)$, Theorem \ref{thm:thm1} requires on the order of $nr\nu\log^{2}\,n$ measurements
which is only $\log(n)$ factor away from the optimal lower bound.  The order of theorem \ref{thm:thm1} is also the same order 
as those used in \cite{gross2011recovering,recht2011simpler}. These works consider the low rank recovery problem with any orthogonal basis and the matrix completion problem respectively. 
Before the proof of main result, we illustrate Theorem~\ref{thm:thm1} on some of the examples discussed in the introduction.  \\*

\noindent
\textbf{Euclidean Distance Geometry Problem}:  A matrix completion formulation and theoretical analysis of the Euclidean distance geometry
problem appears in \cite{tasissa2018exact}. Using the existing analysis in \cite{tasissa2018exact}, $L = \frac{n(n-1)}{2}$, $\lambda_{\max}(\H^{-1}) \le 4$, 
$\lambda_{\min}(\H^{-1}) = \frac{1}{8n}$ and $\|\H^{-1}\|_{\infty}= 8$. The constant $c_v$, which satisfies $c_v \ge \lambda_{\max}(\H^{-1})\|\H^{-1}\|_{\infty}$,
is set to $32$. Using Lemma \ref{corr_props}$(d)$, the correlation condition for the Euclidean Distance Geometry Problem holds with $\mu = 1$.  
Using Theorem \ref{thm:thm1}, it can be seen that the number of samples needed to recover the underlying low rank Gram matrix for the Euclidean distance geometry problem
is $O(nr\nu\log^2 n)$. \\*

\noindent
\textbf{Spectrally Sparse Signal Reconstruction}:  For simplicity, consider the case $n_1=n_2=n$. 
For this problem, the orthonormality of the Hankel basis  implies that $\H = \I$. Therefore, $\lambda_{\max}(\H^{-1}) = \lambda_{\min}(\H^{-1}) = \|\H^{-1}\|_{\infty} = 1$ and $L = n^2$.
The correlation condition holds trivially with $\mu = 0$. The number of samples needed to cover the underlying low rank matrix is $O(nr\nu\log^2 n)$.\\*

\noindent
\textbf{Weighted Nuclear Norm Minimization}:  As discussed earlier, the weight matrix $\D$ is diagonal.
For simplicity, assume that $\sum_{i=1}^{n} \D_{i,i}= 1$. In this case,
the size of a diagonal entry informs the proportion of weight assigned to the corresponding row of the true matrix. For the main result in Theorem \ref{thm:thm1}
to hold, certain assumptions are necessary. First, the basis $\{\D^{-1}\w_{\alpha}\}_{\alpha=1}^{L}$ needs to satisfy the correlation condition.
Second, the size of the maximum and minimum eigenvalues of the matrix $\bar{\H} = (\D^{-1}\W)^{T} (\D^{-1}\W) = \W^T(\D^{-1})^2\W$
is important. After minor analysis, with $\H = \W^T\W$,  $\lambda_{\max}(\bar{\H})\le \lambda_{\max}(\H)+\lambda_{\max}\left(\W^T[(\D^{-1})^2-\I]\W\right)$.
A similar analysis leads to $\lambda_{\min}(\bar{\H}) \ge \lambda_{\min}(\H) + \lambda_{\min} \left(\W^T[(\D^{-1})^2-\I]\W\right)$. 
Assume that, with no weighting, the original basis $\w_{\alpha}$ satisfies the necessary conditions for Theorem \ref{thm:thm1} to hold. 
For the weighted nuclear norm minimization to hold, one choice of sufficient conditions is that $\lambda_{\max}\left(\W^T[(\D^{-1})^2-\I]\W\right)= c_1 \lambda_{\max}(\H)$
and $\lambda_{\min}\left(\W^T[(\D^{-1})^2-\I]\W\right) = c_2\lambda_{\min}(\H)$ where $c_1,c_2$ are dimension-free constants. A given choice of $\D$
can be checked if it verifies these criterion. If the result is positive, the complexity of  $O(nr\nu\log^2 n)$ from Theorem \ref{thm:thm1} can be attained. Note that
the conditions above are sufficient but not necessary. Sharper and more explicit condition on $\D$ requires further analysis and is not within the scope of this paper.
It can be surmised that, in practice, one is working with a fixed basis matrix $\W$ and controlling the spectrum of $\bar{\H}$ in terms of $\D$ is more amenable
to analysis.

\begin{remark}

It should be remarked that the minimum number of samples noted in Theorem \ref{thm:thm1} can be lowered,
the constants could be improved, if one is working with explicit basis. The analysis presented here is generic and does not
assume specific structure of the basis. Where the latter is readily available, most inequalities appearing in the technical
details can be tightened lowering the sample complexity. For instance, if the basis is orthonormal as in the problem of
spectrally sparse signal construction, the analysis in \cite{gross2011recovering} gives tight results. In general, for
explicit basis with some structure, one can adopt the analysis in this paper and improve certain bounds.

\end{remark}

Now we return to the main proof. For ease, the proof is structured into several intermediate results.
The starting result is Theorem \ref{uniqueness_conditions} which 
shows that if certain conditions hold,  $\M$ is a unique  solution to \eqref{eq:m_completion}.

\begin{theorem}\label{uniqueness_conditions}
Given $\X\in \S$, let $\bDelta = \X - \M$ denote the deviation from the true low rank matrix $\M$. 
$\bDelta_{\T}$ and $\bDelta_{\T^{\perp}}$ denote the orthogonal projection of
 $\bDelta$ to $\T$ and $\T^{\perp}$ respectively. For any given $\Omega$ with $|\Omega|=m$, the following two statements hold.
\begin{enumerate}
\item[(a).]
If $\displaystyle \|\bDelta_{\T}\|_{F} \ge \sqrt{2L}\frac{\lambda_{\max}(\H^{-1})}{\lambda_{\min}(\H^{-1})}   
\|\bDelta_{\T^{\perp}}\|_{F}$ and $ \lambda_{\min}\left(\P_{\T}\,\F\,\P_{\T}\right) >\frac{1}{2}\lambda_{\min}(\H)$, then $\Po \bDelta \neq 0$.
\item [(b).] If $ \displaystyle \|\bDelta_{\T}\|_{F} < \sqrt{2L}\frac{\lambda_{\max}(\H^{-1})}{\lambda_{\min}(\H^{-1})}
\|\bDelta_{\T^{\perp}}\|_{F}$ for $\bDelta \in \ker \Po$, and there exists a  $\Y \in \textrm{range}\,\, \Po^{*}$ satisfying,
\begin{equation}\label{eq:yconds}
\|\P_{\T}\Y-\textrm{Sgn}\,\M\|_{F} \le \frac{1}{4}\sqrt{\frac{1}{2L}} \frac{\lambda_{\min}(\H^{-1})}{\lambda_{\max}(\H^{-1})}      \quad \textrm {and} \quad \|\P_{\T^{\perp}}\Y\|\le \frac{1}{2}
\end{equation}
then $\|\X\|_{*}= \|\M + \bDelta\|_{*}> \|\M\|_{*}$.

\end{enumerate}
\end{theorem}

\noindent
Theorem $\ref{uniqueness_conditions}$(a) states that, for ``large" $\bDelta_{\T}$, any deviation from $\M$ is not in the null space of the operator. 
Theorem $\ref{uniqueness_conditions}$(b) states that, for ``small" $\bDelta_{\T}$,  deviations from $\M$ increase the nuclear norm. 
The theorem at hand is deterministic and at this stage no assumptions are made on the construction of the set $\Omega$. As long as the assumptions
of the theorem are satisfied, the theorem will hold true. After proving the theorem, we proceed to argue that the conditions in the theorem hold
with very high probability. This will require certain sampling conditions and a suitable choice of $m = |\Omega|$.

\subsection{Proof of Theorem \ref{uniqueness_conditions}}
\begin{proof} [Proof of Theorem \ref{uniqueness_conditions}(a)] First, observe that $\|\Po\,\bDelta\|_{F} =\|\Po\,\bDelta_{\T} + \Po\,\bDelta_{\T^{\perp}}\|_{F} \ge \| \Po\, \bDelta_{\T}\|_{F}-
\| \Po\,\bDelta_{\T^{\perp}}\|_{F}$. Since we want to show that  $\Po \bDelta \neq 0$, the observation 
leads to considering a lower bound for $\|\Po\, \bDelta_{\T}\|_{F}$ and an upper bound for
$\|\Po\,\bDelta_{\T^{\perp}}\|_{F}$. For any $\X$, $\|\Po\, \X\|_{F}^{2}$
can be bounded as follows.
\begin{equation*}
\|\Po\,\X\|_{F}^{2} = \langle \X\,,\Po^{*}\,\Po\,\X\rangle = 
\frac{L^2}{m^2}\sum_{\beta \in \Omega} \sum_{\alpha \in \Omega}  \langle \X\,,\w_{\alpha}\rangle  \langle \X\,,\w_{\beta}\rangle \langle \z_{\alpha}\,,\z_{\beta}\rangle=  
\frac{L^2}{m^2}\sum_{\beta \in \Omega} \sum_{\alpha \in \Omega}  \langle \X\,,\w_{\alpha}\rangle  \langle \X\,,\w_{\beta}\rangle \H^{\alpha,\,\beta} 
\end{equation*}
where the last inequality uses the fact that $\langle \z_{\alpha}\,,\z_{\beta}\rangle = \H^{\alpha,\,\beta}$. The min-max theorem applied to the above equation results 
\begin{equation}\label{eq:minmax_bound}
\frac{L^2}{m^2}\lambda_{\min}(\H^{-1})\sum_{\alpha \in \Omega} \langle \X\,,\w_{\alpha}\rangle ^{2}\le 
\|\Po\,\X\|_{F}^{2} \le \frac{L^2}{m^2}\lambda_{\max}(\H^{-1})\sum_{\alpha \in \Omega} \langle \X\,,\w_{\alpha}\rangle ^{2}
\end{equation}
Setting $\X = \bDelta_{\T^{\perp}}$ and using the right inequality above, we obtain 
\begin{equation} \label{eq:pdelta_perp_upper_bound}
\|\Po\,\bDelta_{\T^{\perp}}\|_{F}^{2} \le \frac{L^2}{m^2}\lambda_{\max}(\H^{-1})\sum_{\alpha \in \Omega} \langle \bDelta_{\T^{\perp}} \,,\w_{\alpha}\rangle ^{2}\le
m\frac{L^2}{m^2} \frac{\lambda_{\max}(\H^{-1})}{\lambda_{\min}(\H^{-1})}  \|\bDelta_{\T^{\perp}}\|_{F}^{2}
\end{equation}
where the last inequality uses the fact that $\sum_{\alpha \in \Us} \langle \X\,,\w_{\alpha}\rangle^{2} \le \lambda_{\max}(\H)\|\X\|_{F}^{2}$
(Lemma  \ref{norm_equivalent}) and  the constant $m$  bounds the maximum number of repetitions for any given measurement.  Analogously, setting $\X =\bDelta_{\T}$ and using the left inequality 
in \eqref{eq:minmax_bound}, we obtain 
\begin{equation} \label{eq:pdelta_lower_bound_init}
\|\Po\,\bDelta_{\T}\|_{F}^{2} \ge 
  \frac{L^2}{m^2}\lambda_{\min}(\H^{-1}) \sum_{\alpha \in \Omega} \, \langle \bDelta_{\T}\,,\w_{\alpha}\rangle ^{2} 
= \frac{L}{m} \lambda_{\min}(\H^{-1})  \langle \bDelta_{\T}\,,\F\,\bDelta_{\T}\rangle 
  \end{equation}
  The next step considers the projection onto $\T$ of the restricted frame operator and applies the min-max theorem
  resulting the following inequality. 
\begin{align} 
\|\Po\,\bDelta_{\T}\|_{F}^{2} \ge 
    \frac{L}{m}\lambda_{\min}(\H^{-1}) \langle \bDelta_{\T}\,,\F\,\bDelta_{\T}\rangle   
  &=    \frac{L}{m}\lambda_{\min}(\H^{-1})  \langle \bDelta_{\T}\,,\P_{\T}\,\F\,\P_{\T}\,\bDelta_{\T}\rangle  \nonumber  \\
  & \ge 
  \frac{L}{m} \lambda_{\min}(\H^{-1})  \lambda_{\min}(\P_{\T}\,\F\,\P_{\T})\,\|\bDelta_{T}\|_{F}^{2}    \label{eq:pdelta_lower_bound}
  \end{align}
Above, the first equality follows since $\P_{\T}$ is self adjoint and evidently $\bDelta_{\T} \in \T$.
The inequality in \eqref{eq:pdelta_lower_bound} can be reduced further using the assumption $ \lambda_{\min}\left(\P_{\T}\,\F\,\P_{\T}  \right)>\frac{1}{2}\lambda_{\min}(\H)$
in the theorem resulting
\begin{align} 
\|\Po\,\bDelta_{\T}\|_{F}^{2} > \frac{L}{2m} \lambda_{\min}(\H^{-1})\lambda_{\min}(\H) \,\|\bDelta_{T}\|_{F}^{2}    \label{eq:pdelta_lower_bound_simplified}
  \end{align}
Finally, use the inequalities in \eqref{eq:pdelta_perp_upper_bound} and \eqref{eq:pdelta_lower_bound_simplified} and the assumption in the theorem to
show that $\|\Po\,\bDelta\|_{F}>0$ as follows. 
\begin{align*}
\|\Po\,\bDelta\|_{F}  & >  \sqrt{\frac{L}{2m}\lambda_{\min}(\H^{-1})\lambda_{\min}(\H)} \,\|\bDelta_{T}\|_{F}-
\frac{L}{\sqrt{m}} \sqrt{\frac{\lambda_{\max}(\H^{-1})}{\lambda_{\min}(\H^{-1})}}  \|\bDelta_{\T^{\perp}}\|_{F} \\
& \ge  \sqrt{\frac{L}{2m}\lambda_{\min}(\H^{-1})\lambda_{\min}(\H)}              \left(  \sqrt{2L}\frac{\lambda_{\max}(\H^{-1})}{\lambda_{\min}(\H^{-1})} \right)\|\bDelta_{\T^{\perp}}\|_{F}
-\frac{L}{\sqrt{m}} \sqrt{\frac{\lambda_{\max}(\H^{-1})}{\lambda_{\min}(\H^{-1})}}  \|\bDelta_{\T^{\perp}}\|_{F} = 0 
\end{align*}  
This concludes proof of Theorem \ref{uniqueness_conditions}(a).  
  
\end{proof}

 \begin{remark}
 (a) The upper bound estimate for $||\Po\,\bDelta_{\T^{\perp}}||_{F}^{2}$ is not optimal. This is so since 
 the number of times a given measurement can be duplicated is set to $m$ which is a worst case estimate. 
 One approach to improve the estimate is to make use of standard concentration inequalities and argue that
 the expected number of duplicates for a given measurement is much smaller than $m$. A second option,
 as noted in the remark on the sampling model section, is to use a uniform sampling with out replacement model.
 With this, the estimate can be improved since the factor $m$ is no longer necessary. While these alternatives
 lead to a better estimate of the upper bound, the final gain in terms of number of measurements is minor as the improved estimates
 are inside of a $\log_2$. For these reasons and ease of presentation, the current estimate is used in the forthcoming
 analysis.   
(b)Lower bounding the term $ \sum_{\alpha \in \Omega} \, \langle \bDelta_{\T}\,,\w_{\alpha}\rangle ^{2}$ does not
lend itself to simpler analysis. For example, the use of standard concentration inequalities results probability of failures that scale with $n$. 
Since $\|\Po\,\bDelta_{\T}\|_{F}^{2} = \langle \bDelta_{\T}\,,\P_{\T}\Po^{*}\Po\P_{\T}\bDelta_{\T}-\P_{\T}\bDelta_{\T}\rangle+
\|\P_{\T}\bDelta_{\T}\|_{F}^{2}$, equivalently, we can consider an upper bound of $\|\P_{\T}\Po^{*}\Po\P_{\T}\bDelta_{\T}-\P_{\T}\|$ to lower bound $\|\Po\,\bDelta_{\T}\|_{F}^{2}$. 
The upper bound calculations can be carried out but the calculations are involved and assume certain structure of the basis
$\w_{\alpha}$.  The simple alternative approach shown above is general since it does not make restrictive assumptions on the measurement basis. 
 
\end{remark} 

\begin{proof} [Proof of Theorem \ref{uniqueness_conditions}(b)]
Let $\X = \M+\bDelta$ be a feasible solution to \eqref{eq:m_completion} with the condition that 
 $\|\bDelta_{\T}\|_{F} <\sqrt{2L}\frac{\lambda_{\max}(\H^{-1})}{\lambda_{\min}(\H^{-1})}\\ \|\bDelta_{\T^{\perp}}\|_{F}$
and $\bDelta \in \ker \Po$. The goal now is to show that, for any $\X$ that satisfies these assumptions,  the nuclear norm minimization is violated meaning
that $\|\X\|_{*} =\|\M + \bDelta\|_{*}> \|\M\|_{*}$. The proof of this fact makes use of the dual certificate approach
in \cite{gross2011recovering}. The idea is to endow a certain object, named a dual certificate $\Y$, with certain conditions
so as to ensure that any $\X$ satisfying the earlier made assumptions is not a solution to \eqref{eq:m_completion}. 
It then becomes a task to construct the certificate which satisfies the preset conditions. For ease of later reference, we start with the former 
task reproducing a proof, with minor changes, in section $2$E of \cite{gross2011recovering}.  First, using the duality of
the spectral norm and the nuclear norm, note that there exists a $\bm{\Lambda} \in \T^{\perp}$ with $||\bm{\Lambda}||=1$ such
that $\langle \bm{\Lambda}\,, \P_{\T^{\perp}}\,(\bDelta) \rangle=||\P_{\T^{\perp}}\,(\bDelta) ||_{*}$. 
Second, using the characterization of the subgradient of the nuclear norm \cite{watson1992characterization},
$\partial \|\M\|_{*} = \{ \textrm{Sgn} \,\M + \bm{\Gamma}~|~ \bm{\Gamma}\in \T^{\perp} ~\&~\|\P_{\T^{\perp}}\,\bm{\Gamma}\| \le 1\}$. 
With this, it can be readily verified
that $\textrm{Sgn}\, \M + \bm{\Lambda}$ is a subgradient of $||\M||_{*}$ at $\M$. Mathematically, we have 
$ \|\M+\bDelta\|_{*} \ge ||\M||_{*} + \langle  \textrm{Sgn} \,\M + \bm{\Lambda}\,, \bDelta\rangle$. 
Using the condition $\Y \in \textrm{range}\,\,\Po^{*}$ which implies $\langle \Y\,,\bDelta \rangle = 0$ in the previous inequality, we obtain
\begin{align*}
 \|\M+\bDelta\|_{*} & \ge ||\M||_{*} + \langle  \textrm{Sgn} \,\M +  \bm{\Lambda}\,, \bDelta\rangle \\
                                & = ||\M||_{*} + \langle  \textrm{Sgn} \,\M + \bm{\Lambda}-\Y\,, \bDelta\rangle \\          
                                & = ||\M||_{*} + \langle \textrm{Sgn}\,\M - \P_{\T}\,\Y\,,\bDelta_{\T} \rangle + 
||\bDelta_{\T^{\perp}}||_{*} -\langle \P_{\T^{\perp}}\Y\,,\bDelta_{\T^{\perp}} \rangle                
\end{align*}
The third equality follows using the earlier choice of $\bm{\Lambda}$ and the fact that $\textrm{Sgn}\,\M \in \T$ (see Lemma \ref{signx_prop}).
Finally, we apply the assumptions of the theorem to the last equation above to obtain 
\begin{align*}
   \|\M+\bDelta\|_{*}   &\ge ||\M||_{*} + \langle \textrm{Sgn}\,\M - \P_{\T}\,\Y\,,\bDelta_{\T} \rangle + 
||\bDelta_{\T^{\perp}}||_{*} -\langle \P_{\T^{\perp}}\Y\,,\bDelta_{\T^{\perp}} \rangle\\
&\ge  ||\M||_{*}-\frac{1}{4}\sqrt{\frac{1}{2L}} \frac{\lambda_{\min}(\H^{-1})}{\lambda_{\max}(\H^{-1})}  \|\bDelta_{\T}\|_{F}
+\|\bDelta_{\T^{\perp}}\|_{*}-\frac{1}{2}\|\bDelta_{\T^{\perp}}\|_{*} \\
&\ge  ||\M||_{*} -\frac{1}{4}\sqrt{\frac{1}{2L}} \frac{\lambda_{\min}(\H^{-1})}{\lambda_{\max}(\H^{-1})} \|\bDelta_{\T}\|_{F}
+\frac{1}{2}\|\bDelta_{\T^{\perp}}\|_{F}\\
&> ||\M||_{*}  -\frac{1}{4}\sqrt{\frac{1}{2L}} \frac{\lambda_{\min}(\H^{-1})}{\lambda_{\max}(\H^{-1})}
\left(\sqrt{2L}\frac{\lambda_{\max}(\H^{-1})}{\lambda_{\min}(\H^{-1})}   \|\bDelta_{\T^{\perp}}\|_{F}  \right)
+\frac{1}{2}\|\bDelta_{\T^{\perp}}\|_{F}\\
& =||\M||_{*} + \frac{1}{4}\|\bDelta_{\T^{\perp}}\|_{F}
\end{align*}
It can be concluded that $\|\M+\bDelta\|_{*} > \|\M\|_{*}$ as desired. 
\end{proof}

Next, we state and prove a corollary which shows that $\M$ is a unique solution to \eqref{eq:m_completion} if the deterministic assumptions of 
Theorem \ref{uniqueness_conditions} hold. 

\begin{corollary}
If the conditions of Theorem \ref{uniqueness_conditions} hold, $\M$ is a unique solution to \eqref{eq:m_completion}.
\end{corollary}
\begin{proof}
Define $\bDelta = \X-\M$ for any $\X \in \S$. Using Theorem \ref{uniqueness_conditions}(a), $\Po\,\bDelta\neq \bm{0}$ 
if $\displaystyle \|\bDelta_{\T}\|_{F}^{2} \ge \left(\sqrt{2L}\frac{\lambda_{\max}(\H^{-1})}{\lambda_{\min}(\H^{-1})}   
\|\bDelta_{\T^{\perp}}\|_{F}\right)^{2}$. It then suffices to consider the case
$\displaystyle \|\bDelta_{\T}\|_{F}^{2} < \left(\sqrt{2L}\frac{\lambda_{\max}(\H^{-1})}{\lambda_{\min}(\H^{-1})}   
\|\bDelta_{\T^{\perp}}\|_{F}\right)^{2}$ for  $\bDelta \in \ker \Po$. For this case, using the proof of Theorem \ref{uniqueness_conditions}(b),
$\|\X\|_{*}> \|\M\|_{*}$. Therefore, $\M$ is the unique solution to \eqref{eq:m_completion}.
\end{proof}

\subsection{Proof of Theorem \ref{thm:thm1}}

Using the corollary above, if the two conditions in Theorem \ref{uniqueness_conditions} hold, it follows that $\M$
is a unique solution to \eqref{eq:m_completion}. The first condition in Theorem \ref{uniqueness_conditions}(a)
is the assumption that $ \lambda_{\min}\left(\P_{\T}\,\F\,\P_{\T}\right)>\frac{1}{2}\lambda_{\min}(\H)$.
This will ensure that the minimum eigenvalue of the operator $\P_{\T}\,\F\,\P_{\T}$ is bounded away from zero.
Using the operator Chernoff bound in \cite{tropp2012user} restated below, Lemma \ref{min_eig} addresses
the assumption.

\begin{theorem}[Chernoff bound in \cite{tropp2012user}]
Consider a finite sequence $\{\Lc_{k}\}$ of independent, random, self-adjoint operators, acting on matrices in $\real^{n\times n}$, that satisfy
\[
\Lc_{k}\succeq \bm{0} \quad \textrm{ and } \quad \|\Lc_{k}\| \le R \quad \textrm{almost surely}
\] 
Compute the minimum eigenvalue of the sum of the expectations,
\[
\mu_{\min}:=\lambda_{\min} \left(\sum_{k} E[\Lc_{k}]\right) 
\]
Then, we have
\[
 Pr\,\bigg[ \lambda_{\min}\left(\sum_{k}\Lc_{k}\right) \le (1-\delta)\,\mu_{\min}\bigg] \le n\,\bigg[ \frac{\exp(-\delta)}{(1-\delta)^{1-\delta}}  \bigg]^{\frac{\mu_{\min}}{R}}   
\quad \textrm{for } \delta\in [0,1]
\]
\end{theorem}

\noindent
For $\delta \in [0,1]$, using Taylor series of $\log(1-\delta)$, note that $(1-\delta)\log(1-\delta) \ge -\delta+\frac{\delta^2}{2}$. 
This results the following simplified estimate.
\[
Pr\,\bigg[ \lambda_{\min}\left(\sum_{k}\Lc_{k}\right) \le (1-\delta)\,\mu_{\min}\bigg] \le n\,\exp\left(-\delta^2\,\frac{\mu_{\min}}{2R}\right) 
\quad \textrm{for } \delta\in [0,1]
\]

\begin{lemma} \label{min_eig}
Consider the operator $\P_{\T}\,\F\,\P_{\T}: \T \rightarrow \T$. With $\kappa = \frac{mn}{Lr}$, the following estimate holds. 
\[
\textrm{Pr}\,\left( \lambda_{\min}\left(\P_{\T}\,\F\,\P_{\T}\right)\le \frac{1}{2}\lambda_{\min}(\H)\right)
\le  n\exp\left(-\frac{\lambda_{\min}(\H)^{2}\kappa}{8\nu}   \right)
\] 
\end{lemma}

\begin{proof}
Recall the restricted frame operator $\F = \sum_{\alpha \in \Omega} \frac{L}{m}\langle \X\,,\w_{\alpha}\rangle \w_{\alpha}$. 
For $\X \in \T$, $\P_{\T}\,\F\,\P_{\T}\,\X$ can be equivalently represented as follows.
\[
\P_{\T}\,\F\,\P_{\T}\,\X = \sum_{\alphab \in \Omega} \frac{L}{m}\langle \X\,,\P_{\T}\,\w_{\alphab}\rangle\, \P_{\T}\,\w_{\alphab}
\]
Let $\Lc_{\alphab} =\frac{L}{m}\langle \cdot\,,\P_{\T}\,\w_{\alphab}\rangle\, \P_{\T}\,\w_{\alphab}$ denote the operator
in the summand. Since $\Lc_{\alphab}$ is positive semidefinite, the operator Chernoff bound can be used. The bound requires estimate of an upper bound $R$ of the spectrum norm of $\Lc_{\alphab}$ and $\mu_{\min} = \lambda_{\min}(\sum_{\betab\in \Omega} E[\Lc_{\betab}])$. 
First, we estimate $R$ as follows. 
\begin{align*}
\left\|\frac{L}{m}\langle \cdot\,,\P_{\T}\,\w_{\alphab}\rangle\, \P_{\T}\,\w_{\alphab}\right\| = 
 \frac{L}{m}\|\P_{\T}\,\w_{\alphab}\|_{F}^{2} \le \frac{L}{m}\lambda_{\max}(\H^{-1})\frac{\nu r}{n}
\end{align*}
The last inequality follows from the coherence estimate in \eqref{eq:coherencew}. 
With this, set $R = \displaystyle \frac{L}{m}\lambda_{\max}(\H^{-1})\frac{\nu r}{n}$. Next, we consider the estimate of $\mu_{\min}$ by first evaluating $\sum_{\betab\in \Omega} E[\Lc_{\betab}]$.
\[
\sum_{\betab\in \Omega} E[\Lc_{\betab}] = \sum_{\betab \in \Omega} \bigg[ \sum_{\alphab \in\Us} \frac{1}{m}\langle \cdot\,,\P_{\T}\,\w_{\alphab}\rangle\, \P_{\T}\,\w_{\alphab}\bigg]
= \sum_{\alphab\in\Us } \langle \cdot\,,\P_{\T}\,\w_{\alphab}\rangle\, \P_{\T}\,\w_{\alphab}
\]
For any $\X \in \T$,  $\displaystyle \langle \X\,, \sum_{\betab\in \Omega} E[\Lc_{\betab}](\X) \rangle$ can be lower bounded as follows.
\[
\langle \X\,, \sum_{\betab\in \Omega} E[\Lc_{\betab}](\X) \rangle = \sum_{\betab \in\Us} \langle \X,\,\w_{\betab}\rangle ^{2} \ge \lambda_{\min}(\H)\|\X\|_{F}^{2}\]
with the last inequality following from Lemma \ref{norm_equivalent}. The variational characterization of the minimum eigenvalue,
along with the fact that $\sum_{\alphab\in \Omega} E[\Lc_{\alphab}]$ is a self-adjoint operator, implies that the minimum eigenvalue of $ \sum_{\alphab\in \Omega} E[\Lc_{\alphab}] $ is at least 
$\lambda_{\min}(\H)$. 
With this, set $\mu_{\min} = \lambda_{\min}(\H)$. The final step is to apply the operator Chernoff bound 
with $R =\frac{L}{m}\lambda_{\max}(\H^{-1})\frac{\nu r}{n}$ and $\mu_{\min} = \lambda_{\min}(\H)$. Setting $\delta = \frac{1}{2}$, 
$ \lambda_{\min}\left(\P_{\T}\,\F\,\P_{\T}\right)>\frac{1}{2}\lambda_{\min}(\H)$ 
with probability of failure at most $p_1$ given by
\[
p_1 = n\exp\left(-\frac{\lambda_{\min}(\H)^{2}\kappa}{8\nu}   \right)
\]
This concludes the proof. 
\end{proof}

Lemma \ref{min_eig} shows that $\lambda_{\min}\left(\P_{\T}\,\F\,\P_{\T}\right)>\frac{1}{2}\lambda_{\min}(\H)$ holds with
probability at least $1-p_1$ where the probability of failure is at most  $\displaystyle p_1 = n\exp\left(-\frac{\kappa}{8\nu}\right)$ with 
 $\displaystyle \kappa = \frac{mn}{Lr}$. 

In what follows, the conditions in Theorem \ref{uniqueness_conditions}(b) are analyzed. 
The statement there assumes the existence of a certain dual certificate $\Y$ that satisfies
the conditions in \eqref{eq:yconds}. In \cite{gross2011recovering}, David Gross devised a novel scheme, the golfing scheme, to construct the dual certificate $\Y$. 
Before showing the scheme, some notations are in order: 1) The random set $\Omega$ is partitioned into
$l$ batches. The $i$-th batch, denoted $\Omega_{i}$, contains $m_i$ elements with 
$\displaystyle \sum_{i=1}^l m_i = m$.  2) For a given batch, the sampling operator can be defined as follows
$\displaystyle\R_{i} = \frac{L}{m_i} \sum_{\alpha \in \Omega_{i}}{} \langle \X \,, \w_{\alpha} \rangle \z_{\alpha}$.
The inductive scheme is shown below.
\begin{equation} \label{eq:golfing_scheme}
\Q_{0} = \textrm{sgn } \M , \quad \Y_{i} = \sum_{j=1}^{i} \R^{*}_{j}  \Q_{j-1}, \quad \Q_{i} = \textrm{sgn } \M - \P_{\T}\Y_{i}, \quad i = 1,...,l
\end{equation}
The main idea of the remaining analysis is to employ the golfing scheme and certify that the conditions in \eqref{eq:yconds} hold with very high probability.
In the analysis of the golfing scheme, the initial task is to show that the first condition in \eqref{eq:yconds} 
holds. This requires a probabilistic estimate of $\|\P_{\T}\Po^{*}\P_{\T}\X-\P_{\T}\X\|_{F}\ge t$ for a fixed matrix $\X \in \S$
and will be addressed in Lemma \ref{injective} to follow shortly.
The proof of the lemma relies on the vector Bernstein inequality in \cite{gross2011recovering}.
We use a slightly modified version of this inequality which is stated below.

\begin{theorem}[Vector Bernstein inequality] 

Let $\x_1,...,\x_m$ be independent zero-mean vector valued random variables. Assume
that $\underset{i}{\max}\, \|\x_{i}\|_{2}\le R$ and $\sum_{i} E[\|\x_i\|_{2}^{2}]\le \sigma^{2}$.
For any $t\le \frac{\sigma^2}{R}$, the following estimate holds. 
\[
\textrm{Pr}\bigg[\left\|\sum_{i=1}^{m}  \x_i\right\|_{2}\ge t\bigg] \le \exp\left(-\frac{t^2}{8\sigma^2}+\frac{1}{4}\right),
\]

\end{theorem}

\begin{lemma} \label{injective}
Given an arbitrary fixed $\X \in \S$, for $t\le 1 $ with $\kappa = \frac{mn}{Lr}$, the following estimate holds.
\begin{equation}
\textrm{Pr}\,(\|\P_{\T}\Po^{*}\P_{\T}\X-\P_{\T}\X\|_{F}\ge t\|\X\|_{F}) \le 
\exp\left(-\frac{t^2\kappa}{8\left(\lambda_{\max}(\H^{-1})\|\H^{-1}\|_{\infty}\nu+\frac{n}{Lr}\right)}+\frac{1}{4}\right)
\end{equation}

\end{lemma}

\begin{proof}
In what follows, with out loss of generality, we assume $\|\X\|_{F}=1$. Using the dual basis expansion, $\P_{\T} \Po^{*}\P_{\T}\X-\P_{\T}\X$ can be represented as follows.
\begin{equation}\label{eq:ptrpt_pt}
\P_{\T} \Po^{*}  \P_{\T}\X-\P_{\T}\X = \sum_{{\alpha\in\Omega}}\bigg[
\frac{L}{m} \langle \P_{\T}\X\,,\z_{\alpha} \rangle\\P_{\T}\,\w_{\alpha}
-\frac{1}{m}\P_{\T}\X\bigg]
\end{equation}
The summand, denoted $\Y_{\alpha}$, can be written as $\Y_{\alpha} = \X_{\alpha}-E[\X_{\alpha}]$. Since
$E[\Y_{\alpha}]=\bm{0}$, it satisfies the condition for the vector Bernstein inequality and we proceed to consider
appropriate bounds for $\|\Y_{\alpha}\|_{F}$ and $E[\|\Y_{\alpha}\|_{F}^{2}]$.  First, we bound  $\|\Y_{\alpha}\|_{F}$
making use of the coherence conditions \eqref{eq:coherencew} and \eqref{eq:coherencev}. 
\begin{align*}
\|\Y_{\alpha}\|_{F}= \left\|\frac{L}{m}\langle \P_{\T}\X\,,\z_{\alpha} \rangle\, \P_{\T}\,\w_{\alpha}
-\frac{1}{m}\P_{\T}\X\right\|_{F}  &\le
\left\|\frac{L}{m}\langle \P_{\T}\X\,,\z_{\alpha} \rangle \P_{\T}\,\w_{\alpha}\right\|_{F}
+\left\|\frac{1}{m}\P_{\T}\X\right\|_{F} \label{eq:norm_Y_squared}\\
& \le \frac{L}{m}\max_{\alpha\in \Us} \|\P_{\T}\,\w_{\alpha}\|_{F}\max_{\alpha\in \Us} \|\P_{\T}\,\z_{\alpha}\|_{F}
+\frac{1}{m}\nonumber\\
& \le \frac{1}{m}\left(\frac{L}{n}\lambda_{\max}(\H^{-1})\|\H^{-1}\|_{\infty}\nu r+1\right)
\end{align*}
Therefore, set $R = \frac{1}{m}\left(\frac{L}{n}\lambda_{\max}(\H^{-1})\|\H^{-1}\|_{\infty}\nu r+1\right) $. 
To upper bound $E[\|\Y_{\alpha}\|_{F}^{2}]$, we start with the definition of $E[\|\Y_{\alpha}\|_{F}^{2}]$
and proceed as follows. 
\begin{align*}
E[\|\Y_{\alpha}\|_{F}^{2}] & = E\bigg[\frac{L^2}{m^2} \langle \P_{\T}\X\,,\z_{\alpha} \rangle^{2}\|\P_{\T}\,\w_{\alpha}\|_{F}^{2} +\frac{1}{m^2}\|\P_{\T}\X\|_{F}^{2}-
\frac{2L}{m^2}\langle \P_{\T}\X\,,\z_{\alpha}\rangle \langle \P_{\T}\X\,,\w_{\alpha}\rangle  \bigg] \\
& =E\left[\frac{L^2}{m^2} \langle \P_{\T}\X\,,\z_{\alpha} \rangle^{2}\|\P_{\T}\,\w_{\alpha}\|_{F}^{2}\right] +\frac{1}{m^2}\|\P_{\T}\X\|_{F}^{2}-
\frac{2}{m^2}\langle \sum_{\alpha \in \Us} \langle \P_{\T}\X\,,\w_{\alpha}\rangle\z_{\alpha}\,, \P_{\T}\X\rangle   \\
& = E\left[\frac{L^2}{m^2} \langle \P_{\T}\X\,,\z_{\alpha} \rangle^{2}\|\P_{\T}\,\w_{\alpha}\|_{F}^{2}\right] -\frac{1}{m^2}\|\P_{\T}\X\|_{F}^{2}\\
& \le  \frac{L}{m^2} \max_{\alpha\in \Us}\|\P_{\T}\,\w_{\alpha}\|_{F}^{2} \sum_{\alpha\in \Us}\langle \P_{\T}\X\,,\z_{\alpha} \rangle^{2}+\frac{1}{m^2}\\
& \le \frac{L}{m^2}\lambda_{\max}(\H^{-1})^{2}\frac{\nu r}{n} + \frac{1}{m^2} \le  \frac{L}{m^2}\lambda_{\max}(\H^{-1}) \|\H^{-1}\|_{\infty} \frac{\nu r}{n} + \frac{1}{m^2}
\end{align*}
Above, the second inequality results from the coherence conditions \eqref{eq:coherencew} and  \eqref{eq:coherencev} and application of Lemma \ref{norm_equivalent}.
With this, set $\sigma^{2} =  \frac{1}{m}\left( \frac{L}{n}\lambda_{\max}(\H^{-1})\|\H^{-1}\|_{\infty} \nu r+1\right)$.  To conclude the proof, we apply
the vector Bernstein inequality with the specified $R$ and $\sigma$. 
For $t\le \frac{\sigma^{2}}{R} = 1 $, with $\kappa= \frac{mn}{Lr}$, the following estimate holds. 
\begin{equation}
\textrm{Pr}\,(\|\P_{\T}\Po^{*}\P_{\T}\X-\P_{\T}\X\|_{F}\ge t) \le 
\exp\left(-\frac{t^2\kappa}{8\left(\lambda_{\max}(\H^{-1})\|\H^{-1}\|_{\infty}\nu+\frac{n}{Lr}\right)}+\frac{1}{4}\right)
\end{equation}
\end{proof}

Next, it will be argued that the golfing scheme \eqref{eq:golfing_scheme} certifies the conditions
in \eqref{eq:yconds} with very high probability. In particular, we have the following lemma. 

\begin{lemma} \label{construction_of_y}
 $\Y_{l}$ obtained from the golfing scheme \eqref{eq:golfing_scheme} satisfies the conditions
 in \eqref{eq:yconds} with failure probability which is at most 
$\displaystyle p = \sum_{i=1}^{l} p_2(i) + p_3(i) + p_4(i)$
where  
$\displaystyle p_2(i) = \exp\left(-\frac{\kappa_i}{32\left(\lambda_{\max}(\H^{-1})\|\H^{-1}\|_{\infty}\nu+\frac{n}{Lr}\right)}+\frac{1}{4}\right)$, \\
$\displaystyle p_{3}(i) = 2n \exp\left(\frac{-3\min\left((\mu+1) \|\H^{-1}\|_{\infty},\frac{1}{4}\right)^{2}\kappa_i}{8(\mu+1) \|\H^{-1}\|_{\infty}^{2}\nu}\right)$ and
$\displaystyle p_{4}(i) = n^{2} \exp \left(-\frac{3\kappa_i}{32\left(c_v\nu+\frac{n}{Lr}\right)}\right)$ 
with $\displaystyle k_{i} = \frac{m_in}{Lr}$.

\end{lemma}

\begin{proof}
In what follows, we repeatedly make use of the fact that $\Q_{i} \in \T$ since $\textrm{sgn } \M \in \T$ from Lemma \ref{signx_prop}. The main idea for showing that the first condition in \eqref{eq:yconds} holds
relies on a recursive form of $\Q_i$ which can be derived as follows. 
\begin{align}
  \Q_{i} & =  \textrm{sgn } \M - \P_{\T} \left(\sum_{j=1}^{i} \R^{*}_{j} \Q_{j-1}\right)= 
	\textrm{sgn } \M - \P_{\T} \left(\sum_{j=1}^{i-1} \R^{*}_{j} \Q_{j-1} + \R^{*}_{i}\Q_{i-1}\right)\nonumber\\
  & = \textrm{sgn } \M - \P_{\T}  \sum_{j=1}^{i-1} \R^{*}_{j} \Q_{j-1} -\P_{\T}  \R^{*}_{i}\Q_{i-1}
  =  \textrm{sgn } \M - \P_{\T}  \Y_{i-1} -\P_{\T}  \R^{*}_{i}\Q_{i-1}\nonumber\\
   & = \Q_{i-1} -\P_{\T}  \R^{*}_{i}\Q_{i-1}= (\P_{\T} - \P_{\T} \R^{*}_{i} \P_{\T} ) \Q_{i-1} \label{eq:qi_update}
\end{align}
The first condition in \eqref{eq:yconds} is a bound on $\|(\P_{\T} - \P_{\T} \R^{*}_{l} \P_{\T})\Q_{l-1}\|_{F}= \|\Q_{l}\|_{F}$.
Using Lemma \ref{injective} with $t_{2,i} =\frac{1}{2}$,  $\|(\P_{\T} - \P_{\T} \R^{*}_{i} \P_{\T})\Q_{i-1}\|< t_{2,i}\|\Q_{i-1}\|_{F}$
holds with failure probability at most 
\[
p_2(i) = \exp\left(-\frac{\kappa_i}{32\left(\lambda_{\max}(\H^{-1})\|\H^{-1}\|_{\infty}\nu+\frac{n}{Lr}\right)}+\frac{1}{4}\right)
\]
where $\displaystyle \kappa = \frac{m_in}{Lr}$.  Using the recursive formula in \eqref{eq:qi_update} repeatedly, $\|\Q_{i}\|_{F}$  can be upper bounded as follows.
\begin{equation}\label{eq:qi_norm}
  \|\Q_{i}\|_{F} <  \left(\prod_{k=1}^{i} t_{2,k} \right) \|\Q_0\|_F = 2^{-i} \sqrt{r}
\end{equation}
Setting $\displaystyle l = \log_{2}\left(4\sqrt{2L}\frac{\lambda_{\max}(\H)}{\lambda_{\min}(\H)}\sqrt{r}\right)$, 
the first condition in \eqref{eq:yconds} is now satisfied. It can be concluded that, using
a union bound on the failure probabilities $p_2(i)$, $\displaystyle \|\Q_l\|_{F} < \sqrt{r}2^{-l} = \frac{1}{4}\sqrt{\frac{1}{2L}} \frac{\lambda_{\min}(\H^{-1})}{\lambda_{\max}(\H^{-1})}$
holds with failure probability that is at most  $\sum_{i=1}^{l} p_2(i)$. This implies that the first condition in \eqref{eq:yconds} also holds
with the same failure probability.

Next, we consider the second condition in \eqref{eq:yconds} which requires a bound on $\|\P_{\T^{\perp}} \Y_l\|$. 
First, note that $ \|\P_{\T^{\perp}} \Y_{l} \|  
= \| \P_{\T^{\perp}} \sum_{j=1}^{l} \R^{*}_j \Q_{j-1}\|= \|  \sum_{j=1}^{l} 
\P_{\T^{\perp}} \R^{*}_j \Q_{j-1}\| \le   \sum_{j=1}^{l} \|\P_{\T^{\perp}} \R^{*}_j \Q_{j-1}\| $.
As such, in what follows, the focus will be finding a suitable bound for  $\|\P_{\T^{\perp}} \R^{*}_j \Q_{j-1}\|$.
This will be analyzed in Lemma \ref{golfing_size}. A key element in the proof of Lemma \ref{golfing_size} is an assumption 
on the size of  $\max_{\beta} \,|\langle \Q_{i}\,,\z_{\beta} \rangle|$. For ease of notation
in further analysis, let $\eta(\Q_i)$ be defined as: $\eta(\Q_i)=\max_{\beta} \,|\langle \Q_{i}\,,\z_{\beta} \rangle|$. 
The assumption is that, at the $i$-th step of the golfing scheme, $\displaystyle \eta(\Q_{i})^{2}\le \|\H^{-1}\|_{\infty}^{2} \frac{\nu}{n^2} c_{i}^{2}$
where $c_i^2$ is an upper bound for  $\|\Q_{i}\|_{F}^{2}$, $\|\Q_{i}\|_{F}^{2}\le c_{i}^{2}$. The idea is to argue that this assumption holds with very high probability. 
To show this, assume that $\eta(\Q_i)\le t_{4,i}$ with failure probability $p_4(i)$. Setting $t_{4,i} = \frac{1}{2}\eta(\Q_{i-1})$
and applying the inequality $\eta(\Q_i)\le t_{4,i}$ recursively results
\[
\eta(\Q_{i})^{2} \le 2^{-2}\eta(\Q_{i-1})^{2}\le 2^{-2i} \eta(\textrm{sgn}\,\M)^2 \le  2^{-2i} \|\H^{-1}\|_{\infty}^{2}  \frac{\nu r}{n^2} = 
\|\H^{-1}\|_{\infty}^{2} \frac{\nu}{n^2} (2^{-2i} r)
\]
where the last inequality follows from the coherence estimate in \eqref{eq:jointcoherence}.
It can now be concluded that, noting  \eqref{eq:qi_norm}, the inequality above ensures
that $\displaystyle \eta(\Q_{i})^{2}\le \|\H^{-1}\|_{\infty}^{2} \frac{\nu}{n^2} c_{i}^{2}$ with $c_i = 2^{-i}\sqrt{r}$. 
The failure probability $p_{4}(i)$ follows from Lemma \ref{joint_coherence_size_golfing}, noting that $\eta(\Q_{i}) =\eta(\Q_{i-1}-\P_{\T}\R_{i}^{*}\Q_{i-1})$,
and is given by
\[
p_4(i) = n^{2} \exp\left(-\frac{3\kappa_j}{32\left(c_v\nu+\frac{n}{Lr}\right)}\right) \quad \quad \forall i\in[1,l]
\]
Having justified the assumption on the size of $\eta(\Q_i)$, a key part of Lemma \ref{joint_coherence_size_golfing}, 
we now consider certifying the second condition in  \eqref{eq:yconds}.  Assume that 
$\|\P_{\T^{\perp}}\R^{*}_j\Q_{j-1}\|< t_{3,j}c_{j-1}$, with $\|\Q_{j-1}\|_{F}\le c_{j-1}=2^{-(j-1)}$, holds with failure probability $p_3(j)$. 
Fixing $\displaystyle t_{3,j}=\min\left(\frac{(\mu+1) \|\H^{-1}\|_{\infty}}{\sqrt{r}},\frac{1}{4\sqrt{r}}\right)$, Lemma \ref{golfing_size} gives
\[
p_{3}(j) = 2n \exp\left(\frac{-3\min\left((\mu+1) \|\H^{-1}\|_{\infty},\frac{1}{4}\right)^{2}\kappa_j }{8(\mu+1) \|\H^{-1}\|_{\infty}^{2}\nu}\right)
\]
where $\displaystyle \kappa_j = \frac{m_jn}{Lr}$.  
$\|\P_{\T^{\perp}} \Y_{l}\|$ can now be upper bounded as follows.
\begin{align*}
  \|\P_{\T^{\perp}} \Y_{l}\| \le  \sum_{k=1}^{l} \|\P_{\T^{\perp}}\R^{*}_{k}\Q_{k-1}\|
  < \sum_{k=1}^{l} t_{3,k}c_{k-1} \le \frac{1}{4\sqrt{r}} \sum_{k=1}^{l} c_{k-1}
  = \frac{1}{4\sqrt{r}} \sum_{k=1}^{l} \sqrt{r}2^{-(k-1)} < \frac{1}{2}
  \end{align*}
Applying the union bound over the failure probabilities, $\|\P_{\T^{\perp}} \Y_{l}\|<\frac{1}{2}$ holds true with failure probability 
which is at most $\sum_{j=1}^{l}[p_3(j)+p_4(j)]$. With the same failure probability, the second condition in \eqref{eq:yconds} holds true.

\end{proof}

Lemma \ref{golfing_size} will be proved shortly using the Bernstein inequality in  \cite{tropp2015introduction}
which is restated below for convenience.

\begin{theorem}[Bernstein inequality]

  Consider a finite sequence $\{\X_i\}$ of independent, random matrices with dimension $n$.
  Assume that 
  \[
E[\X_i] =0 \quad \textrm{and} \quad \|\X_{i}\|\le R \quad \forall i
\]
Let the matrix variance statistic of the sum $\sigma^{2}$ be defined as
\[
\sigma^{2} = \max\left(\left\|\sum_{i}E[\X_{i}^{T}\X_{i}]\right\|,\left\|\sum_{i}E[\X_{i}\X_{i}^{T}]\right\|\right)
\]
For all $t\ge 0$,
\begin{equation} \label{eq:bern}
  Pr\bigg[\left\|\sum_{i}\X_{i}\right\|>t\bigg] \le
\begin{cases}
\displaystyle 2n \exp\left(-\frac{3t^2}{8\sigma^{2}}\right) & t \le \frac{\sigma^{2}}{R} \vspace{0.2cm}\\
 \displaystyle 2n \exp\left(-\frac{3t}{8R}\right)            & t\ge \frac{\sigma^{2}}{R} 
\end{cases}
\end{equation}
\end{theorem}

\begin{lemma}  \label{golfing_size}
  Consider a fixed matrix $\G \in \T$. Assume that $\max_{\beta}\,\langle \G\,,\z_{\beta}\rangle^{2} \le \|\H^{-1}\|_{\infty}^{2} \frac{\nu}{n^2}\,c^{2}$
  with $c$ set as $\|\G\|_{F}^{2}\le c^{2}$. Then, with $\displaystyle \kappa_j = \frac{m_j n}{Lr}$, the following estimate holds for all 
  $\displaystyle t\le \frac{(\mu+1) \|\H^{-1}\|_{\infty}}{\sqrt{r}}$.
  \[
\textrm{Pr}\,(\|\P_{\T^{\perp}}\R^{*}_{j} \G\| \ge t\,c) \le 2n \exp\left(-\frac{3t^2\kappa_{j}r}{8(\mu+1) \|\H^{-1}\|_{\infty}^{2}\nu}\right)
\]
\end{lemma}

\begin{proof}[Proof of Lemma \ref{golfing_size}]
Using the dual basis representation,  $\displaystyle \P_{\T^{\perp}}\R^{*}_{j} \G = \sum_{\alpha\in\Omega_j} \frac{L}{m_j}  \langle \G\,,\z_{\alpha} \rangle \P_{\T^{\perp}} \w_{\alpha}$.
The summand, denoted $\X_{\alpha}$,  has zero expectation since $\G \in \T$. With this, 
the zero mean assumption for Bernstein inequality is satisfied. Next, we consider suitable estimates for $R$ and
$\sigma^{2}$. The latter necessitates an estimate for $\max(\|E[\X_{\alpha}\X_{\alpha}^{T}]\|,\|E[\X_{\alpha}^T\X_{\alpha}]\|)$. 
First, we bound  $\|E[\X_{\alpha}\X_{\alpha}^{T}]\|$. 
Since $\displaystyle \X_{\alpha}\X_{\alpha}^{T} = \frac{L^2}{m_{j}^{2}}\langle \G\,,\z_{\alpha} \rangle ^{2} (\P_{\T^{\perp}} \w_{\alpha})(\P_{\T^{\perp}} \w_{\alpha})^{T}$, 
using  Lemma \ref{operator_norm_of_sum} and the fact that $\w_{\alpha}\w_{\alpha}^{T}$ is positive semidefinite, 
 $\|E[\X_{\alpha}\X_{\alpha}^{T}]\|$ can be upper bounded as follows.
$$
  \left\|E[\X_{\alpha}\X_{\alpha}^{T}]\right\| \le  \frac{L}{m_j^{2}}\, \underset{\|\varphi\|_2=1}{\max}\sum_{\alpha\in \Us} 
	\langle \G\,,\z_{\alpha} \rangle ^{2} 
  \langle \varphi\,, \w_{\alpha}\w_{\alpha}^{T}\varphi \rangle 
\le  \frac{L}{m_{j}^2} \, \underset{\alpha\in \Us}{\max}\,\, \langle \z_{\alpha}\,,\G\rangle ^{2}\,
\underset{\|\varphi\|_2=1}{\max} \,  \langle \varphi\,, \left(\sum_{\alpha\in \Us} \w_{\alpha}\w_{\alpha}^{T} \right)\varphi \rangle \label{eq:golfing_var_inequality}
$$
Using the correlation condition, in particular Lemma \ref{corr_props}(b), $\|\sum_{\alpha\in \Us} \w_{\alpha}\w_{\alpha}^{T}\|\le (\mu+1) n$, we obtain 
$$\left\|E[\X_{\alpha}\X_{\alpha}^{T}]\right\| \le \frac{L\,(\mu+1)\,n}{m_j^2}\,\underset{\alpha\in \Us}{\max}\,\, \langle \z_{\alpha}\,,\G\rangle ^{2}
\le \frac{L\,(\mu+1)\,n}{m_{j}^2} \|\H^{-1}\|_{\infty}^{2} \frac{\nu}{n^2}\,c^{2} $$
An analogous calculation and reasoning as above yields $\left\|E[\X_{\alpha}^{T}\X_{\alpha}]\right\| \le \frac{L\,(\mu+1)n}{m_{j}^2} \|\H^{-1}\|_{\infty}^{2} \frac{\nu}{n^2}\,c^{2} $.
Therefore, using triangle inequality, set 
\[
\sigma^{2} = \frac{L\,(\mu+1)}{m_{j}} \|\H^{-1}\|_{\infty}^{2} \frac{\nu}{n}\,c^{2}
\]
To complete the proof, it remains to estimate $R$. 
\begin{equation}
  \|\X_{\alpha}\| \le \frac{L}{m_{j}} |\underset{\alpha\in \Us}{\max}\,\, \langle \z_{\alpha}\,,\G\rangle|\, \|\P_{\T^{\perp}}\w_{\alpha}\|
  \le \frac{L}{m_{j}} \|\H^{-1}\|_{\infty} \frac{\sqrt{\nu}}{n}\,c
\end{equation}
Two pertinent cases have to be considered.  If $\displaystyle \nu\ge \frac{1}{r}$,  
$\displaystyle\|\X_{\alpha}\|  \le \frac{L}{m_{j}} \|\H^{-1}\|_{\infty} \frac{\nu\sqrt{r}}{n}\,c=R_{1}$
and if $\displaystyle\nu < \frac{1}{r}$, $\displaystyle\|\X_{\alpha}\|  \le \frac{L}{m_{j}} \|\H^{-1}\|_{\infty} \frac{\sqrt{\nu}}{n}\,c=R_{2}$. 
Note that $\displaystyle\frac{\sigma^{2}}{R_1}=  \frac{(\mu+1) \|\H^{-1}\|_{\infty}c}{\sqrt{r}}$
and $\displaystyle\frac{\sigma^{2}}{R_2} = \frac{(\mu+1) \|\H^{-1}\|_{\infty}c}{\sqrt{r}}$.
To conclude the proof, we apply the Bernstein inequality. 
An application of the Bernstein inequality results the following estimate.
\begin{equation}
\textrm{Pr}(\|\P_{\T^{\perp}}\R^{*}_{j} \G\| \ge t) \le 2n \exp\left(-\frac{3t^2\kappa_{j}r}{8(\mu+1) \|\H^{-1}\|_{\infty}^{2}\nu c^2}\right)
\end{equation}
 for all $\displaystyle t \le \frac{(\mu+1) \|\H^{-1}\|_{\infty}c}{\sqrt{r}}$ with $\displaystyle \kappa_{j}= \frac{m_{j}n}{Lr}$. 
This concludes the proof of Lemma \ref{golfing_size}.
\end{proof}

Next, it will be argued that $\M$ is a unique solution to \eqref{eq:m_completion} with very high probability. 
The argument considers two separate cases based on comparing
$\|\bDelta_{\T}\|_{F}^{2}$ and $\|\bDelta_{\T^{\perp}}\|_{F}^{2}$. This motivates us to define the following two sets:
$\S_{1}=\bigg\{\X\in \S: \|\bDelta_{\T}\|_{F}^{2} \ge \left(\sqrt{2L} \frac{\lambda_{\max}(\H^{-1})}{\lambda_{\min}(\H^{-1})}\|\bDelta_{\T^{\perp}}\|_{F}\right)^{2} \bigg\}$ and 
$\S_2=\bigg\{\X\in \S:\|\bDelta_{\T}\|_{F}^{2} < \left(\sqrt{2L} \frac{\lambda_{\max}(\H^{-1})}{\lambda_{\min}(\H^{-1})}\|\bDelta_{\T^{\perp}}\|_{F}\right)^{2}\,\&\,\Po(\bDelta)=\bm{0} \bigg\}$. With this, assuming that $\Omega$ is sampled uniformly at random with replacement, the two cases are as follows.

\begin{enumerate}
\item For all $\X\in\S_{1}$, set $|\Omega| = m$ ``sufficiently large" such that $\textrm{Pr}\, \left(\left\{\Omega\subset \Us\,\big|\, |\Omega| = m, \lambda_{\min}\left(\P_{\T}\,\F\,\P_{\T}\right)> \frac{\lambda_{\min}(\H)}{2}\right\}\right) $ $\ge 1-p_1$ based on Lemma \ref{min_eig}. 
Therefore,  all $\X\in\S_{1}$ are feasible solutions to \eqref{eq:m_completion} with probability at most $p_1$ from Theorem \ref{uniqueness_conditions}(a).

\item For all $\X\in\S_{2}$, we obtain 
\[
\textrm{Pr}\,\left(\left\{\Omega\subset \Us\, \big|\,\, |\Omega| = m, \Y \in \textrm{range}\,\,\Po^{*}\, \& \, \|\P_{\T}\Y-\textrm{Sgn}\,\M\|_{F} \le \frac{1}{4}\sqrt{\frac{1}{2L}} \frac{\lambda_{\min}(\H^{-1})}{\lambda_{\max}(\H^{-1})}\, \&\, \|\P_{\T^{\perp}}\Y\|\le \frac{1}{2}\right\}\right)\ge 1-\epsilon
\] 
with $\displaystyle \epsilon=  \sum_{i=1}^{l} \left[p_{2}(i)+p_3(i)+p_4(i)\right]$ by 
setting $|\Omega| = m$  ``sufficiently large" based on Lemma \ref{construction_of_y}.  Then, the probability of all $\X\in\S_{2}$ being solutions to
\eqref{eq:m_completion} is at most $\epsilon$ from Theorem \ref{uniqueness_conditions}(b).

\end{enumerate}

Using the above two cases and employing the union bound, any $\X\in \S$ different from $\M$ is a solution to the general matrix completion problem with probability 
at most $p= p_1+\sum_{i=1}^{l}[p_{2}(i)+p_3(i)+p_4(i)]$. In the arguments above, we have  used the terms ``sufficiently large'',
``small probability'' and ``very high probability'' with out being precise. The goal now is to set everything explicit. 
First, define very high probability as a probability of at least $1-n^{-\beta}$ for $\beta>1$. 
Analogously, define small failure probability as a probability of at most $n^{-\beta}$ for some $\beta>1$. 
To recover the underlying matrix with high probability, the idea is to carefully set the remaining free parameters
$m$, $l$ and $m_i$  so that the $p \le n^{-\beta}$ for $\beta>1$. 
This necessitates revisiting all the failure probabilities in the analysis. $p_1$, the first failure probability, is
the probability that the condition in Theorem \ref{uniqueness_conditions}(a) does not hold.
In the construction of the dual certificate $\Y$ via the golfing scheme, three failure probabilities, $p_{2}(i)$, $p_{3}(i)$ and $p_{4}(i)$, 
$\forall i\in[1,l]$, appear.
With this, all the failure probabilities are noted below. 
\begin{align*}
  p_1 &=  n\exp\left(-\frac{\lambda_{\min}(\H)^{2}\kappa}{8\nu}   \right) \quad ; \quad
 p_2(i) =   \exp\left(-\frac{\kappa_i}{32\left(\lambda_{\max}(\H^{-1})\|\H^{-1}\|_{\infty}\nu+\frac{n}{Lr}\right)}+\frac{1}{4}\right)\\
 p_{3}(i) &= 2n \exp\left(\frac{-3\min\left((\mu+1) \|\H^{-1}\|_{\infty},\frac{1}{4}\right)^{2}\kappa_i}{8(\mu+1) \|\H^{-1}\|_{\infty}^{2}\nu}\right)\quad ; \quad
p_{4}(i) = n^{2}\exp \left(-\frac{3\kappa_i}{32\left(c_v\nu+\frac{n}{Lr}\right)}\right)
\end{align*}
To suitably set $m_i$,  since $k_i= \frac{m_i n}{Lr}$, we can set $k_i$ with the condition that all the failure probabilities are at most $\displaystyle \frac{1}{4l}n^{-\beta}$ for
$\beta>1$. A minor calculation results one suitable choice of $k_i$, $k_i = 48\big(C\nu+\frac{1}{nr}\big)\big(\beta\log(n)+\log(4l)\big)$, with $C$ defined as
\begin{equation}\label{eq:c}
C  = \max\left(\lambda_{\max}(\H^{-1})\|\H^{-1}\|_{\infty}, c_{v}, \frac{(\mu+1) \|\H^{-1}\|_{\infty} }{\min\left((\mu+1)\|\H^{-1}\|_{\infty},\frac{1}{4}\right)^{2}}\right)
\end{equation}
The total failure probability, applying the union bound, is bounded above by $n^{-\beta}$. The number of
measurements, $m= lrnk_i$, is at least
\begin{equation} \label{eq:m_general_basis_expression}
\log_{2}\left(4\sqrt{2L}\frac{\lambda_{\max}(\H)}{\lambda_{\min}(\H)}\sqrt{r}\right)nr
\left(48\big[C\nu+\frac{n}{Lr}\big]\big[\,\beta\log(n)+\log\left(4\log_{2}\left(4\sqrt{2L}\frac{\lambda_{\max}(\H)}{\lambda_{\min}(\H)}\sqrt{r}\right)\right)\big]\right)
\end{equation}
This finishes the proof of Theorem \ref{thm:thm1}. It can be concluded that the minimization
program in $\eqref{eq:m_completion}$ recovers the underlying matrix with very high probability.

\subsection{Noisy General Matrix Completion}
In practical applications, the measurements in the general matrix completion problem are prone to noise. 
This motivates the analysis of the robustness of the nuclear norm minimization program for the general matrix completion problem. 
In particular, consider the following noisy general matrix completion problem where noise is modeled as
Gaussian noise with mean $\mu$ and variance $\sigma$.
\begin{align}
&\underset{\X\in \S}{\textrm{minimize }} \quad \|\X\|_{*} \nonumber\\
&\textrm{subject to }\quad \|\Po(\X) - \Po(\M)\|\le \delta  \label{eq:m_completion_noisy}
\end{align}
Above, $\delta$ characterizes the level
of noise. In \cite{candes2010matrix}, under certain assumptions, the authors show the robustness of the nuclear norm
minimization algorithm for the matrix completion problem. Following this existing analysis and the dual basis
framework, we expect that a robustness result, such as the one below, can be attained. 

\begin{theorem}\label{thm1_noisy}
Let $\M\ \in \real^{n\times n}$ be a matrix of rank $r$ that obeys the coherence conditions
\eqref{eq:coherencew1}, \eqref{eq:coherencev1} and \eqref{eq:jointcoherence1}
with coherence $\nu$ and satisfies the correlation condition \eqref{eq:correlation} with correlation parameter $\mu$.
Define $C$ as follows: 
$\displaystyle C  = \max\left(\lambda_{\max}(\H^{-1})\|\H^{-1}\|_{\infty}, c_{v}, \frac{(\mu+1) \|\H^{-1}\|_{\infty} }{\min\left((\mu+1)\|\H^{-1}\|_{\infty},\frac{1}{4}\right)^{2}}\right) $.
Assume $m$ measurements $\langle \M\,,\w_{\alphab}\rangle$, sampled uniformly at random with replacement, 
are corrupted with Gaussian noise of mean $\mu$, variance $\sigma$ and noise level $\delta$. 
For $\beta>1$, if
\begin{equation} 
m \ge \log_{2}\left(4\sqrt{2L}\frac{\lambda_{\max}(\H)}{\lambda_{\min}(\H)}\sqrt{r}\right)nr
\left(48\big[C\nu+\frac{n}{Lr}\big]\big[\,\beta\log(n)+\log\left(4\log_{2}\left(4\sqrt{2L}\frac{\lambda_{\max}(\H)}{\lambda_{\min}(\H)}\sqrt{r}\right)\right)\big]\right)
\end{equation}
then
\[
\|\M-\bar{\M}\|_{F} \le f\left(n,\frac{m}{n^2},\delta\right)
\]
where $\bar{\M}$ is a solution to \eqref{eq:m_completion_noisy} with probability at least $1-n^{-\beta}$.

\end{theorem}

\noindent
\textbf{Remark}: A proof of the robustness theorem above requires specifying the level of noise
and making the function $f$ explicit. We leave this as a future work. 

\section{Conclusion}
\label{sec:conclusion}
In this paper, we study the problem of recovering a low rank matrix given a few of its expansion coefficients
with respect to any basis. The considered problem generalizes existing analysis for the standard matrix completion problem and
low rank recovery problem with respect to an orthonormal basis. The main analysis uses the dual basis approach and 
is based on dual certificates.  An important assumption in the analysis is a proposed sufficient condition on the basis
matrices named as the correlation condition. This condition can be checked in $O(n^3)$ computational time
and holds in many cases of deterministic basis matrices where the restricted isometry property (RIP) condition might not hold or is NP-hard to verify. 
If this condition holds and the underlying low rank matrix obeys the coherence condition with parameter $\nu$,  under additional mild
assumptions, our main result shows that the true matrix can be recovered with very high probability from $O(nr\nu\log^2n)$ uniformly random
sampled coefficients. Future research will consider a detailed analysis of the correlation condition and evaluating the effectiveness of
the framework of the general matrix completion in certain applications. 

\section{Acknowledgment}

Abiy Tasissa would like to thank Professor David Gross for correspondence over email regarding the work in \cite{gross2011recovering}. Particularly, the proof of
Lemma $10$ is a personal communication from Professor David Gross.

\bibliographystyle{plain}
\bibliography{matrix_completion}

\begin{thebibliography}{10}

\bibitem{ahlswede2002strong}
Rudolf Ahlswede and Andreas Winter.
\newblock Strong converse for identification via quantum channels.
\newblock {\em IEEE Transactions on Information Theory}, 48(3):569--579, 2002.

\bibitem{bandeira2013certifying}
Afonso~S Bandeira, Edgar Dobriban, Dustin~G Mixon, and William~F Sawin.
\newblock Certifying the restricted isometry property is hard.
\newblock {\em IEEE transactions on information theory}, 59(6):3448--3450,
  2013.

\bibitem{biswas2006semidefinite}
Pratik Biswas, Tzu-Chen Lian, Ta-Chung Wang, and Yinyu Ye.
\newblock Semidefinite programming based algorithms for sensor network
  localization.
\newblock {\em ACM Transactions on Sensor Networks (TOSN)}, 2(2):188--220,
  2006.

\bibitem{cai2017fast}
Jian-Feng Cai, Tianming Wang, and Ke~Wei.
\newblock Fast and provable algorithms for spectrally sparse signal
  reconstruction via low-rank hankel matrix completion.
\newblock {\em Applied and Computational Harmonic Analysis}, 2017.

\bibitem{candes2015phase}
Emmanuel~J Candes, Yonina~C Eldar, Thomas Strohmer, and Vladislav Voroninski.
\newblock Phase retrieval via matrix completion.
\newblock {\em SIAM review}, 57(2):225--251, 2015.

\bibitem{candes2010matrix}
Emmanuel~J Candes and Yaniv Plan.
\newblock Matrix completion with noise.
\newblock {\em Proceedings of the IEEE}, 98(6):925--936, 2010.

\bibitem{candes2009exact}
Emmanuel~J Cand{\`e}s and Benjamin Recht.
\newblock Exact matrix completion via convex optimization.
\newblock {\em Foundations of Computational mathematics}, 9(6):717--772, 2009.

\bibitem{candes2006robust}
Emmanuel~J Cand{\`e}s, Justin Romberg, and Terence Tao.
\newblock Robust uncertainty principles: Exact signal reconstruction from
  highly incomplete frequency information.
\newblock {\em IEEE Transactions on information theory}, 52(2):489--509, 2006.

\bibitem{candes2013phaselift}
Emmanuel~J Candes, Thomas Strohmer, and Vladislav Voroninski.
\newblock Phaselift: Exact and stable signal recovery from magnitude
  measurements via convex programming.
\newblock {\em Communications on Pure and Applied Mathematics},
  66(8):1241--1274, 2013.

\bibitem{candes2005decoding}
Emmanuel~J Candes and Terence Tao.
\newblock Decoding by linear programming.
\newblock {\em IEEE transactions on information theory}, 51(12):4203--4215,
  2005.

\bibitem{candes2010power}
Emmanuel~J Cand{\`e}s and Terence Tao.
\newblock The power of convex relaxation: Near-optimal matrix completion.
\newblock {\em IEEE Transactions on Information Theory}, 56(5):2053--2080,
  2010.

\bibitem{ding2010sensor}
Yichuan Ding, Nathan Krislock, Jiawei Qian, and Henry Wolkowicz.
\newblock Sensor network localization, euclidean distance matrix completions,
  and graph realization.
\newblock {\em Optimization and Engineering}, 11(1):45--66, 2010.

\bibitem{fang2013using}
Xingyuan Fang and Kim-Chuan Toh.
\newblock Using a distributed sdp approach to solve simulated protein molecular
  conformation problems.
\newblock In {\em Distance Geometry}, pages 351--376. Springer, 2013.

\bibitem{fazel2001rank}
Maryam Fazel, Haitham Hindi, and Stephen~P Boyd.
\newblock A rank minimization heuristic with application to minimum order
  system approximation.
\newblock In {\em American Control Conference, 2001. Proceedings of the 2001},
  volume~6, pages 4734--4739. IEEE, 2001.

\bibitem{fazel2003log}
Maryam Fazel, Haitham Hindi, and Stephen~P Boyd.
\newblock Log-det heuristic for matrix rank minimization with applications to
  hankel and euclidean distance matrices.
\newblock In {\em American Control Conference, 2003. Proceedings of the 2003},
  volume~3, pages 2156--2162. IEEE, 2003.

\bibitem{foygel2011learning}
Rina Foygel, Ohad Shamir, Nati Srebro, and Ruslan~R Salakhutdinov.
\newblock Learning with the weighted trace-norm under arbitrary sampling
  distributions.
\newblock In {\em Advances in Neural Information Processing Systems}, pages
  2133--2141, 2011.

\bibitem{glunt1993molecular}
W~Glunt, TL~Hayden, and M~Raydan.
\newblock Molecular conformations from distance matrices.
\newblock {\em Journal of Computational Chemistry}, 14(1):114--120, 1993.

\bibitem{gross2011recovering}
David Gross.
\newblock Recovering low-rank matrices from few coefficients in any basis.
\newblock {\em Information Theory, IEEE Transactions on}, 57(3):1548--1566,
  2011.

\bibitem{gross2010note}
David Gross and Vincent Nesme.
\newblock Note on sampling without replacing from a finite collection of
  matrices.
\newblock {\em arXiv preprint arXiv:1001.2738}, 2010.

\bibitem{hoeffding1963probability}
Wassily Hoeffding.
\newblock Probability inequalities for sums of bounded random variables.
\newblock {\em Journal of the American statistical association},
  58(301):13--30, 1963.

\bibitem{kalofolias2014matrix}
Vassilis Kalofolias, Xavier Bresson, Michael Bronstein, and Pierre
  Vandergheynst.
\newblock Matrix completion on graphs.
\newblock {\em arXiv preprint arXiv:1408.1717}, 2014.

\bibitem{kueng2014ripless}
Richard Kueng and David Gross.
\newblock Ripless compressed sensing from anisotropic measurements.
\newblock {\em Linear Algebra and its Applications}, 441:110--123, 2014.

\bibitem{lai2017solve}
Rongjie Lai and Jia Li.
\newblock Solving partial differential equations on manifolds from incomplete
  interpoint distance.
\newblock {\em SIAM Journal on Scientific Computing}, 39(5):A2231--A2256, 2017.

\bibitem{oymak2011simplified}
Samet Oymak, Karthik Mohan, Maryam Fazel, and Babak Hassibi.
\newblock A simplified approach to recovery conditions for low rank matrices.
\newblock In {\em Information Theory Proceedings (ISIT), 2011 IEEE
  International Symposium on}, pages 2318--2322. IEEE, 2011.

\bibitem{recht2011simpler}
Benjamin Recht.
\newblock A simpler approach to matrix completion.
\newblock {\em The Journal of Machine Learning Research}, 12:3413--3430, 2011.

\bibitem{recht2010guaranteed}
Benjamin Recht, Maryam Fazel, and Pablo~A Parrilo.
\newblock Guaranteed minimum-rank solutions of linear matrix equations via
  nuclear norm minimization.
\newblock {\em SIAM review}, 52(3):471--501, 2010.

\bibitem{srebro2010collaborative}
Nathan Srebro and Ruslan~R Salakhutdinov.
\newblock Collaborative filtering in a non-uniform world: Learning with the
  weighted trace norm.
\newblock In {\em Advances in Neural Information Processing Systems}, pages
  2056--2064, 2010.

\bibitem{tasissa2018exact}
Abiy Tasissa and Rongjie Lai.
\newblock Exact reconstruction of euclidean distance geometry problem using
  low-rank matrix completion.
\newblock {\em To appear, IEEE Transaction on Information Theory, arXiv
  preprint arXiv:1804.04310}, 2018.

\bibitem{tenenbaum2000global}
Joshua~B Tenenbaum, Vin De~Silva, and John~C Langford.
\newblock A global geometric framework for nonlinear dimensionality reduction.
\newblock {\em science}, 290(5500):2319--2323, 2000.

\bibitem{tropp2012user}
Joel~A Tropp.
\newblock User-friendly tail bounds for sums of random matrices.
\newblock {\em Foundations of computational mathematics}, 12(4):389--434, 2012.

\bibitem{tropp2015introduction}
Joel~A Tropp et~al.
\newblock An introduction to matrix concentration inequalities.
\newblock {\em Foundations and Trends{\textregistered} in Machine Learning},
  8(1-2):1--230, 2015.

\bibitem{trosset1997applications}
Michael~W Trosset.
\newblock Applications of multidimensional scaling to molecular conformation.
\newblock 1997.

\bibitem{watson1992characterization}
G~Alistair Watson.
\newblock Characterization of the subdifferential of some matrix norms.
\newblock {\em Linear algebra and its applications}, 170:33--45, 1992.

\end{thebibliography}

\appendix
\section{Appendix $A$}

\begin{lemma} \label{completeness}
Given an orthonormal basis $\{\C_{\alpha}\}_{\alpha=1}^{L=n^2}$, we have $\sum_{\alpha} \C_{\alpha}^{T}\C_{\alpha} = \sum_{\alpha} \C_{\alpha}\C_{\alpha}^{T} = n\I$.
\end{lemma}
\begin{proof}

Consider a matrix $\X\in \real^{n\times n}$ expanded in the orthonormal basis as  $ \X = \sum_{\alpha} \langle \X\,, \C_{\alpha} \rangle \C_{\alpha}$.
With $\x$ and $\c_{\alpha}$ as column-wise vectorized forms of  $\X$ and $\C_{\alpha}$ respectively,
$\x = \sum_{\alpha} \langle \x\,,\c_{\alpha}\rangle \c_{\alpha}= (\sum_{\alpha} \c_{\alpha}\c_{\alpha}^{T})\x$.
This implies that $\sum_{\alpha} \c_{\alpha}\c_{\alpha}^{T}=\I_{n^2}$ where the subscript denotes the size of
the identity matrix.  This is the standard completeness relation. The implication of this relation is that
$\sum_{\alpha} \C_{\alpha}(i,j)  \C_{\alpha}(s,t)  = \delta^{i,j}_{s,t}$.  Next, consider
$\sum_{\alpha} \C_{\alpha}\C_{\alpha}^{T}$. The $(i,j)$-th entry of this sum is given by 
\[
\left(\sum_{\alpha} \C_{\alpha}\C_{\alpha}^{T}\right)_{i,j} = \sum_{\alpha} \sum_{s=1}^{n} \C_{\alpha}(i,s) \C_{\alpha}(j,s) 
= \sum_{s=1}^{n} \sum_{\alpha}  \C_{\alpha}(i,s) \C_{\alpha}(j,s)  = \sum_{s=1}^{n} \delta^{i,s}_{j,s}= n \delta_{i,j}
\]
It follows that  $\sum_{\alpha} \C_{\alpha}\C_{\alpha}^{T} = n\I$. An analogous calculation results $\sum_{\alpha} \C_{\alpha}^{T}\C_{\alpha} = n\I$. 

\end{proof}

\begin{lemma} \label{signx_prop}
If $\X \in \T$, $\textrm{Sgn}\,\X \in \T$.
\end{lemma}
\begin{proof}
Consider the singular value decomposition of $\X$ as $\X = \bm{U} \Sigma \bm{V}^{T}$.
$\textrm{Sgn}\,\X$ is simply $\textrm{Sgn}\,\X = \bm{U}(\textrm{Sgn}\,\Sigma) \bm{V}^{T} = \bm{U} \bm{D} \bm{V}^{T}$
where $\D$ is the diagonal matrix resulting from applying the
sign function to $\Sigma$. Using this decomposition, we consider $\P_{\T^{\perp}} \textrm{sgn }\, \X$.
\begin{align*}
\P_{\T^{\perp}}\, \textrm{sgn } \X & = \textrm{sgn } \X - \P_{\T} \,\textrm{sgn} \X\\
                                & = \bm{U}\bm{D}\bm{V}^T - [ \P_{\U} \,\textrm{sgn } \X +  \textrm{sgn } \X \P_{\V}
                                    - \P_{\bm{U}} \, \textrm{sgn } \X \,\P_{\V} ] \\
                                & = \bm{U}\bm{D}\bm{V}^T - [ \bm{U}\bm{U}^T\bm{U}\bm{D}\bm{V}^T + \bm{U}\bm{D}\bm{U}^T \bm{V}\bm{V}^T - \bm{U}\bm{U}^T \bm{U}
                                \bm{D}\bm{U}^T \bm{V}\bm{V}^T ] = \bm{0}
\end{align*}
Above, the last step follows from the fact that $\bm{U}^T \bm{U} = \I$. It can be concluded that $\textrm{sgn } \X \in \T$.

\end{proof}

\begin{lemma}\label{norm_equivalent}
Given any $\X \in \real^{n\times n}$, the following norm inequalities hold. 
\[
\lambda_{\min}(\H)\,\|\X\|_{F}^{2}\le \sum_{\alphab\in \Us}  \langle \X\,,\w_{\alpha}\rangle^{2} \le \lambda_{\max}(\H)\|\X\|_{F}^{2} \,\,\,;
\,\,\,
\lambda_{\min}(\H^{-1})\,\|\X\|_{F}^{2}\le \sum_{\alphab\in \Us}  \langle \X\,,\z_{\alpha}\rangle^{2} \le \lambda_{\max}(\H^{-1})\|\X\|_{F}^{2} 
\]
\end{lemma}

\begin{proof}
Vectorize the matrix $\X$ and each dual basis $\z_{\alpha}$. It follows that
\[ 
\sum_{\alpha\in \Us} \langle \X\,,\z_{\alpha}\rangle^{2} = \sum_{\alpha\in \Us}  \x^{T}\z_{\alpha}\z_{\alpha}^{T} \x 
= \x^{T} \Z\Z^{T} \x
\]
Orthogonalize $\Z$ with $\overline{\Z} = \Z(\sqrt{\H^{-1}})^{-1}$. Since $\sum_{\beta\in \Us}  \langle \X\,,\z_{\beta}\rangle^{2}   
= \x^{T}\, \overline{\Z}\H^{-1}\overline{\Z}^{T} \x$, we obtain
\[
\lambda_{\min}(\H^{-1})\,\|\x\|_{2}^{2}= \lambda_{\min}(\H^{-1})\,\|\X\|_{F}^{2} \le 
\sum_{\beta\in \Us}  \langle \X\,,\z_{\beta}\rangle^{2} \le \lambda_{\max}(\H^{-1})\,\|\x\|_{2}^{2} = \lambda_{\max}(\H^{-1})\,\|\X\|_{F}^{2}
\]
The above result follows from a simple application of the min-max theorem. An analogous argument gives
\[
\lambda_{\min}(\H)\ \|\X\|_{F}^{2}\le \sum_{\alpha\in \Us}  \langle \X\,,\w_{\alpha}\rangle^{2} \le \lambda_{\max}(\H^{-1})\\|\X\|_{F}^{2}
\]
This concludes the proof.
\end{proof}

\begin{lemma} \label{joint_coherence_size_golfing}
	Define $\displaystyle \eta(\X) = \max_{\beta\in\Us} \,|\langle \X\,,\z_{\beta} \rangle|$.
  For a fixed $\X$ in $\T$, with $\displaystyle\kappa_{j} = \frac{m_{j}n}{Lr}$, the following estimate holds for all $t\le \eta(\X) $. 
  \begin{equation}
    \textrm{Pr}\,(\max_{\beta\in\Us}\,\,|\langle \P_{\T}\R^{*}_{j}\X-\X\,,\z_{\beta}\rangle|\ge t) \le n^{2} \exp\left(-\frac{3t^2\kappa_j}{8\eta(\X)^{2}\left(c_v\nu+\frac{n}{Lr}\right)}\right)
  \end{equation}

  \end{lemma}

\begin{proof}

$\langle \P_{\T}\R^{*}_{j}\X-\X\,,\z_{\beta}\rangle$, for some $\beta$, can be represented in the dual basis as follows.
$$
\langle \P_{\T}\R^{*}_{j}\X-\X\,,\z_{\beta}\rangle = 
	\langle \sum_{\alpha\in\Omega_j} \frac{L}{m_j}\langle \X\,,\z_{\alpha}\rangle \P_{\T}\,\w_{\alpha}-\X\,,\z_{\beta}\rangle
  = \sum_{\alpha\in\Omega_j} \left(\frac{L}{m_j}\langle \X\,,\z_{\alpha}\rangle \langle \P_{\T}\,\w_{\alpha}\,,\z_{\beta}
\rangle -\frac{1}{m_j} \langle \X\,,\z_{\beta}\rangle\right)
$$
The summand, denoted $Y_{\alpha}$, is of the form $X_{\alpha}-E[X_{\alpha}]$ and automatically satisfies $E[Y_{\alpha}]=0$. 
Bernstein inequality can now be applied with appropriate bound on $|Y_{\alpha}|$ and $|E[Y_{\alpha}^{2}]|$. 
First, we bound $|Y_{\alpha}|$ making use of the coherence conditions \eqref{eq:coherencew} and \eqref{eq:coherencev}.
\begin{align*}
  |Y_{\alpha}| &= \bigg|\frac{L}{m_j}\langle \X\,,\z_{\alpha}\rangle \langle \P_{\T}\,\w_{\alpha}\,,\z_{\beta}\rangle -\frac{1}{m_j} 
\langle \X\,,\z_{\beta}\rangle\bigg|
  \le \frac{L}{m_j} \eta(\X)\ \lambda_{\max}(\H^{-1})\|\H^{-1}\|_{\infty}\frac{\nu r}{n}+\frac{1}{m_j}\eta(\X)\\
  & = \frac{1}{m_j}\eta(\X)\left(\frac{L}{n} \lambda_{\max}(\H^{-1})\|\H^{-1}\|_{\infty}\nu r+1\right) 
  \end{align*}
To bound $E[Y_{\alpha}^{2}]$, noting that $E[Y_{\alpha}^{2}] = E[X_{\alpha}^{2}]-E[X_{\alpha}]^{2}$, it follows that
  \begin{align*}
E[Y_{\alpha}^{2}]  &\le  E\bigg[\frac{L^2}{m_{j}^2}\langle \X\,,\z_{\alpha}\rangle^{2} \langle \P_{\T}\,\w_{\alpha}\,,\z_{\beta}\rangle^{2}\bigg]
+\frac{1}{m_j^2} \langle \X\,,\z_{\beta}\rangle^{2}\\
& \le  \frac{L}{m_{j}^2}\sum_{\alpha\in \Us} \langle \X\,,\z_{\alpha}\rangle^{2} \langle \P_{\T}\,\w_{\alpha}\,,\z_{\beta}\rangle^{2}+\frac{1}{m_j^2}\eta(\X)^{2}\\
& \le  \eta(\X)^{2} \frac{L}{m_{j}^2}\sum_{\alpha\in \Us} \langle \P_{\T}\,\w_{\alpha}\,,\z_{\beta}\rangle^{2} +\frac{1}{m_j^2}\eta(\X)^{2}\\
      &\le \eta(\X)^{2} \frac{L}{m_{j}^2}\frac{c_v\nu r}{n}  +\frac{1}{m_j^2}\eta(\X)^{2}
\end{align*}
The last inequality results from the coherence condition in \eqref{eq:coherencev1}.  
To conclude the proof, we apply the Bernstein inequality with $|Y_{\alpha}|\le R =\frac{1}{m_j}\eta(\X)\left(\frac{L}{n} \lambda_{\max}(\H^{-1})\|\H^{-1}\|_{\infty}\nu r+1\right)  $ 
and $\sigma^2= \frac{1}{m_j}\eta(\X)^{2}\left( \frac{L}{n} c_v\nu r + 1\right)$.  With $k_j = \frac{m_jn}{Lr}$,
for $t\le \frac{\sigma^{2}}{R}= \eta(\X)^{2} \frac{\frac{L}{n}c_v\nu r+1}{\frac{L}{n} \lambda_{\max}(\H^{-1})\|\H^{-1}\|_{\infty}\nu r+1}\ge \eta(\X)$, it holds that
\begin{equation}
    \textrm{Pr}(|\langle \P_{\T}\R^{*}_{j}\X-\X\,,\z_{\beta}\rangle|\ge t) \le  \exp\left(-\frac{3t^2\kappa_j}{8\eta(\X)^{2}\left(c_v\nu+\frac{n}{Lr}\right)}\right)
  \end{equation}
Lemma \ref{joint_coherence_size_golfing} now follows from applying a union bound over all elements of the dual basis.
\end{proof}

\begin{lemma}\label{operator_norm_of_sum}
Let $\c_{\alpha}\ge 0$. Then, the following two inequalities hold. 
\[
\left\|\sum_{\alpha} \c_{\alpha}\,(\P_{\T^{\perp}}\,\w_{\alpha}) (\P_{\T^{\perp}}\,\w_{\alpha})^{T}\right\| 
\le \left\|\sum_{\alpha} \c_{\alpha}\, \w_{\alpha}\w_{\alpha}^{T}\right\|
\]
\[
\left\|\sum_{\alpha} \c_{\alpha}\,(\P_{\T^{\perp}}\,\w_{\alpha})^{T} (\P_{\T^{\perp}}\,\w_{\alpha})\right\| 
\le \left\|\sum_{\alpha} \c_{\alpha}\, \w_{\alpha}^{T}\w_{\alpha}\right\|
\]
\end{lemma}

\begin{proof}
We start with the first statement. 
Using the definition of $\P_{\T^{\perp}}\,\w_{\alphab}$, 
$\left\|\sum_{\alpha} \c_{\alpha}\,  (\P_{\T^{\perp}}\,\w_{\alpha}) (\P_{\T^{\perp}}\,\w_{\alpha})^{T}\right\|$ can be written
as follows.
\[
\left\|\sum_{\alpha} \c_{\alpha}\,(\P_{\T^{\perp}}\,\w_{\alpha}) (\P_{\T^{\perp}}\,\w_{\alpha})^{T}\right\| = 
\left\|\sum_{\alpha} \c_{\alpha}\, \P_{\U^{\perp}}\,\w_{\alpha}\,\P_{\V^{\perp}}\,\w_{\alpha}^{T}\,\P_{\U^{\perp}}\right\|
= \left\|\P_{\U^{\perp}}\left( \sum_{\alpha} \c_{\alpha}\, \w_{\alpha}\,\P_{\V^{\perp}}\,\w_{\alpha}^{T}\right)\P_{\U^{\perp}}\right\|
\]
Using the fact that the operator norm is unitarily invariant and $\|\P\X\P\|\le \X$ for any $\X$ and a projection $\P$, $\quad$
$\left\|\sum_{\alpha} \c_{\alpha}\, (\P_{\T^{\perp}}\,\w_{\alpha}) (\P_{\T^{\perp}}\,\w_{\alpha})^{T}\right\|$ can be upper bounded as follows
\begin{align*}
\left\|\sum_{\alpha} \c_{\alpha}\, (\P_{\T^{\perp}}\,\w_{\alphab}) (\P_{\T^{\perp}}\,\w_{\alpha})^{T}\right\| 
\le \left\|\sum_{\alpha} \c_{\alpha}\, \w_{\alpha}\,\P_{\V^{\perp}}\,\w_{\alpha}^{T}\right\|
&= \left\|\sum_{\alpha} \c_{\alpha}\, \w_{\alpha}\,\left(\w_{\alpha}^{T}-\P_{\V}\,\w_{\alpha}^{T}\right)\right\|\\
&= \left\|\sum_{\alpha} [\c_{\alpha}\, \w_{\alpha}\w_{\alpha}^{T}-\c_{\alpha}\,\w_{\alpha}\,\P_{\V}\,\w_{\alpha}^{T}]\right\|
\end{align*}
where the first equality follows from the relation $\P_{\V^{\perp}} = \I -\P_{\V}$. Since $\c_{\alpha}\ge 0$, $\sum_{\alpha} \c_{\alpha}\, \w_{\alpha}\w_{\alpha}^{T}$ is positive semidefinite.
Using the relation $\P_{\V}^{2} = \P_{\V}$ and the assumption that $\c_{\alpha}\ge 0$, 
$\sum_{\alpha} \c_{\alpha}\w_{\alpha}\,\P_{\V}\,\w_{\alpha}^{T}=\sum_{\alpha} \c_{\alpha} \w_{\alpha}\,\P_{\V}\,\P_{\V}\,\w_{\alpha}^{T}$ is 
also positive semidefinite. A similar argument concludes that $\sum_{\alpha} \c_{\alpha}\, \w_{\alpha}\,\P_{\V^{\perp}}\,\w_{\alpha}^{T}$ is also
positive semidefinite. Finally, using the norm inequality, $\|\bm{A}+\bm{B}\|\ge \max(\|\bm{A}\|,\|\bm{B}\|)$, for
positive semidefinite matrices $\bm{A}$ and $\bm{B}$, it can be seen the first statement holds.  An analogous proof
as above yields the second statement concluding the proof. 
\end{proof}

\end{document}